\documentclass[a4paper,UKenglish,cleveref, autoref, thm-restate,numberwithinsect]{lipics-v2021}

\pdfoutput=1 %
\hideLIPIcs  %

\usepackage{amsmath}
\usepackage{amsthm}
\usepackage{amssymb}
\usepackage{hyperref}

\usepackage{xspace}
\usepackage{caption}
\usepackage{xcolor}
\usepackage{tikz}
\usetikzlibrary{positioning}
\usepackage{booktabs}
\usepackage[capitalise]{cleveref}

\usepackage{algorithm}
\usepackage{algpseudocode}
\algrenewcommand\algorithmicrequire{\textbf{Initialize:}}
\algrenewcommand\algorithmicensure{\textbf{At every step $t$:}}
\algnewcommand\algorithmicmatch{\textbf{The conditions that trigger a match and counter resetting:}}
\algnewcommand\MATCH{\item[\algorithmicmatch]}
\algnewcommand\algorithmicpriority{\textbf{Priority rule:}}
\algnewcommand\PRIORITY{\item[\algorithmicpriority]}

\newcommand{\mA}{\mathcal{A}}
\newcommand{\SOL}{\mathrm{SOL}}
\newcommand{\cost}{\mathrm{cost}}
\newcommand{\m}{\mathrm{m}}
\newcommand{\ms}{\mathrm{ms}}
\newcommand{\LEAF}{\mathsf{LEAF}}
\newcommand{\ROOT}{\mathsf{ROOT}}
\newcommand{\READY}{\mathsf{READY}}
\newcommand{\ALG}{\mathrm{ALG}}
\newcommand{\bu}{\mathbf{u}}
\newcommand{\kv}{\mathcal{K}^{V}}
\newcommand{\OPT}{\mathrm{OPT}}
\allowdisplaybreaks

\usepackage[appendix=append, bibliography=common]{apxproof}
\theoremstyle{plain}
\newtheoremrep{claimapp}[theorem]{Claim}
\theoremstyle{plain}
\newtheoremrep{theoremapp}[theorem]{Theorem}
\theoremstyle{plain}
\newtheoremrep{lemmaapp}[theorem]{Lemma}
\title{Online Matching with Size-Based and Convex Delays} 

\author{Junhao Gan}{School of Computing and Information Systems, The University of Melbourne, Australia}{junhao.gan@unimelb.edu.au}{https://orcid.org/0000-0001-9101-1503}{Supported in part by the Australian Government through the Australian Research Council ARC DP230102908.}

\author{Xiao Sun}{School of Computing and Information Systems, The University of Melbourne, Australia}{xiao.sun.1@student.unimelb.edu.au}{https://orcid.org/0009-0003-2635-8764}{Supported by the Australian Government through the Australian Research Council DP240101353 and by the University of Melbourne through the Melbourne Research Scholarship.}

\author{Seeun William Umboh}{School of Computing and Information Systems, The University of Melbourne, Australia \and ARC Training Centre in Optimisation Technologies, Integrated Methodologies, and Applications (OPTIMA), Australia}{william.umboh@unimelb.edu.au}{https://orcid.org/0000-0001-6984-4007}{Supported by the Australian Government through the Australian Research Council DP240101353.}

\authorrunning{J. Gan, X. Sun, S. W. Umboh}

\Copyright{Junhao Gan, Xiao Sun, Seeun William Umboh} %



\nolinenumbers %

\begin{document}

\maketitle

\begin{abstract}
We study the online min-cost perfect matching with delay (MPMD) problem where $m$ requests arrive in a metric space of $n$ points. In MPMD, an algorithm can choose to match a request or to delay, and the objective is to minimise the sum of connection and delay costs. The connection cost of a match is the distance between the locations of two matched requests in the metric, and the increase of the delay cost is a function of the set of unmatched requests at every moment. In this paper, we study two different types of delay functions, size-based (MPMD-Size) and convex delays (MPMD-Convex). 

The study of MPMD-Size was initiated by Deryckere and Umboh (APPROX/RANDOM 2023) where the
instantaneous delay increment is a non-negative monotone function of the number of unmatched requests. We give an exponential improvement on the lower bounds for the deterministic and randomized competitive ratios. We also give improved upper bounds in terms of $n$, as opposed to Deryckere and Umboh's upper bounds that were functions of $m$, which can be much larger than $n$. Our results settle the deterministic competitive ratio (up to constants). At the heart of these results is a succinct encoding scheme of MPMD-Size on a given $n$-point metric as a metrical task system problem on a $2^{n-1}$-point metric. 

We also consider MPMD-Convex proposed by Liu et al. (ISAAC 2018) where the delay cost incurred by each request is a uniform convex delay function of the time difference between its arrival time and the moment that it is matched by the algorithm. They focused on delay functions $f$ that are unbounded, non-decreasing, continuous, and satisfy $f(0)=f'(0)=0$, and showed that the deterministic competitive ratio is  $\Omega(n)$ for $n$-point uniform metrics. We show that, surprisingly, when $f$ is a non-negative, monotone polynomial with $f'(0)>0$, there is an $O(1)$-competitive deterministic algorithm for uniform metrics. Our result completes our understanding of MPMD-Convex on uniform metrics for a broad class of functions. 
\end{abstract}

\section{Introduction}

In recent years, there has been much interest in revisiting classical online problems through the lens of online problems with delay (e.g.~\cite{DBLP:conf/esa/AzarCKT20,DBLP:journals/disopt/Epstein21,DBLP:conf/icalp/ChenKU22}). %
In such problems, instead of serving every request immediately, one can choose to delay the service of a request while accumulating a delay cost, and the total cost of a solution is the cost of serving all requests plus the accumulated delay costs. In such problems, the trade-off between service cost and delay cost is at the center of the decision making of an online algorithm: the algorithm can wait until many requests have arrived and then apply an offline algorithm to serve them cheaply as a batch, but this will incur a large delay cost.

In this paper, we study the Min-Cost Perfect Matching with Delay (MPMD) problem proposed by Emek et al.~\cite{DBLP:conf/stoc/EmekKW16}. In MPMD, there is an underlying metric space of $n$ points in which $m$ requests arrive over time at different points. At every moment when there is an available request, the online algorithm can choose to either match it with another request and incur a \emph{connection cost} equal to the distance between the locations of two requests, or to leave it unmatched and incur a \emph{delay cost}. 

MPMD has many intuitive interpretations and real-life applications. For example, suppose that one needs to design a mechanism for an online game platform which always needs to match two players to compete against each other. The mechanism wants to match two players who have similar levels to each other as this could reduce their dissatisfaction, which is a connection cost that the platform wants to minimize. The platform is allowed to choose not to match a player immediately but to wait for another player whose level is close enough to the first player to join the game and then to match them, so that the connection cost can be reduced. However, if one player has been waiting too long, they might lose patience and become more and more unsatisfied with this game platform, which will incur more delay cost instead. Therefore, the online algorithm needs to balance the trade-off between these two types of costs, which is at the center of studies in this area. This kind of problems are inherently online since we do not know what will come in the future, and once we have made a decision, it is irreversible.

\subparagraph{Delay costs.} While the connection cost of a match is the distance between the locations of the two requests that are matched, there are various delay models according to how the delay cost increases at every moment when there is at least one request pending. Some delay models apply a uniform delay function that the total delay of each request is a function of the time different between its arrival and being matched, e.g. linear~\cite{DBLP:conf/stoc/EmekKW16}, convex~\cite{DBLP:conf/isaac/LiuPWW18}, and concave delays~\cite{DBLP:conf/soda/AzarRV21}. The most general framework is that the delay increment is a function of the set of requests that are pending at every moment and the algorithm is non-clairvoyant that it does not know future delay functions, e.g. set and size-based delays~\cite{DBLP:conf/approx/DeryckereU23}. 

\subparagraph{Competitive ratios: $n$ vs $m$.} Prior studies on these delay models have achieved results that express competitive ratios as functions of $n$ or $m$. Note that $m$ can be incomparably larger than $n$ when the request sequence is sufficiently long, and such algorithms can perform badly even on a simple one-point metric when $m$ tends to infinity (e.g. \cite{DBLP:journals/mst/AzarF20}). Thus, we aim at designing algorithms whose competitive ratios depend only on $n$.

\subsection{Our Contributions}

In this paper, we make contributions to the study of two generalizations of linear delay functions mentioned above that are much less understood than the others: size-based and convex delays.

\subparagraph{Size-based delay (MPMD-Size).} In the size-based delay setting, at every step the delay increment is a monotone function of the number of currently pending requests (See \cref{section 3.1} for a formal definition). We prove bounds on the deterministic and randomized competitive ratios in terms of $n$. Previously, only upper bounds in terms of $m$ are known.

\begin{theorem}\label{det}
For every $n$-point metric, the deterministic competitive ratio of MPMD-Size is $\Theta(2^n)$.
\end{theorem}

\begin{theorem}\label{ran}
For every $n$-point metric, the randomized competitive ratio (against an oblivious adversary) of MPMD-Size is $O(n^2)$ and $\Omega(n)$.
\end{theorem}    

\noindent We emphasize that the lower bounds in \cref{det,ran} are universal lower bounds: they apply to \emph{every} metric space.

\begin{table}[h]
\centering
\begin{tabular}{lcc}
\toprule
Competitive Ratio & Results from~\cite{DBLP:conf/approx/DeryckereU23}  & Our Results \\
\midrule
Deterministic & $O(2^m)$ and $\Omega(n)$& $\Theta(2^n )$ \\
Randomized & $O(m^4)$ and $\Omega(\log n)$ & $O(n^2)$ and $\Omega(n)$ \\
\bottomrule
\end{tabular}
\caption{Comparisons between the results from~\cite{DBLP:conf/approx/DeryckereU23} and our new results for MPMD-Size. The lower bounds of \cite{DBLP:conf/approx/DeryckereU23} only apply to uniform metrics while ours apply to every metric.}\label{table}
\end{table}

\cref{det,ran} improve upon previous results on size-based delays by Deryckere and Umboh~\cite{DBLP:conf/approx/DeryckereU23} in two ways. First, the previous upper bounds for both deterministic and randomized algorithms are expressed as functions of $m$, the number of requests, while our results are expressed as functions of $n$, the number of points in the given metric space, and thus is an improvement in the more interesting regime where $m \gg n$. Second, we also provide a reduction from MTS to MPMD-Size while theirs is only from MPMD-Size to MTS, which exponentially improves lower bounds for both deterministic and randomized algorithms on \emph{every} $n$-point metric. In contrast, their lower bound only holds for uniform metrics. Our results settle the deterministic competitive ratio, up to a constant. See \cref{table} for a detailed comparison between our results and those by Deryckere and Umboh.

\subparagraph{Convex delay (MPMD-Convex).} We also consider the convex delays on $n$-point uniform metrics, first studied by Liu et al.~\cite{DBLP:conf/isaac/LiuPWW18}. Convex delay is applicable to a lot of real-life scenarios as well, e.g., a gaming platform where players think it is acceptable to wait for a short period yet a longer delay can be more and more frustrating and even affect the player's willingness to join the platform again. In this problem, there is a non-decreasing delay function $f$ associated with every request and the delay cost of request $r$ incurred until $t$ is $f(t-a(r))$, where $a(r)$ is its arrival time, and the total delay cost is the sum of every request's individual delay. The previous study by Liu et al.~focused on the setting where $f$ is non-decreasing, continuous, unbounded and satisfy $f(0)=f'(0)=0$. They gave a deterministic lower bound of $\Omega(n)$ and a matching upper bound when $f$ is also a monomial.  

We consider polynomials $f$ such that $f'(0)>0$. Prior to our work, the only known result is that when $f$ is linear, the deterministic competitive ratio is $O(1)$ \cite{DBLP:conf/soda/AzarCK17}. Our second main result fills this gap and completes our understanding of the landscape of the competitive ratio of MPMD under uniform convex monotone polynomial delays on uniform metrics.  

\begin{theorem}[Informal; see \cref{thm-convex}]\label{convex intro}
For any fixed convex monotone polynomial delay $f$ such that $f'(0)>0$, there is an $O(1)$-competitive deterministic algorithm for MPMD under delay $f$ on all $n$-point uniform metrics. 
\end{theorem}

Surprisingly, for different delay functions, their local derivatives at $t=0$ can totally change the competitiveness, even when they have the same superlinear asymptotic behavior when $t\rightarrow\infty$. For example, the deterministic competitive ratio is $\Omega(n)$ for $f(t)=t^2$, yet a minor perturbation to $f(t)=(t+\varepsilon)^2-\varepsilon^2$ by any small $\varepsilon>0$ completely changes it to $O(1)$. We will see in the proof that it is due to the approximability between $f(t)$ and the linear delay $g(t)=t$: as long as $f(0)=f'(0)=0$, the ratio between $f$ and $g$ tends to $0$ when $t$ tends to $0$ and so we cannot use $g$ to approximate $f$ within positive constant factors. The $O(1)$-competitive result applies to a broader class of functions beyond polynomials as well, even some exponential functions; see \cref{morefunctions} for further discussion.

\subsection{Our Techniques}\label{technique}

\subparagraph{MPMD-Size.}

We employ a similar high-level approach as Deryckere and Umboh~\cite{DBLP:conf/approx/DeryckereU23} of encoding MPMD-Size as a \emph{metrical task system} (MTS) problem (defined in \cref{MTS}). Our key contribution is a significantly more efficient encoding.  
 
The main idea is to define the MTS states so that matching requests in the MPMD-Size instance corresponds to a state transition in the MTS instance, and at each timestep, the instantaneous delay increment in the MPMD-Size instance can be realized by a cost vector in the MTS instance. Deryckere and Umboh used the obvious encoding in which there is an MTS state for each even-sized subset of the requests that have arrived so far, and the current MTS state is the set of requests that have been matched so far. While this makes it trivial to encode the instantaneous delay increment as an MTS cost vector---the cost of a state is exactly the instantaneous delay increment of the requests that have arrived but are not matched---the total number of MTS states is the number of all possible even-sized subsets of the requests, which is exponential in $m$. It also complicates matters when applying MTS algorithms as the MTS state space grows as new requests arrive.

Our key insight that enables a more succinct encoding of MPMD-Size as MTS is the following. It never hurts to match two pending requests if they are co-located and so the number of unmatched requests at each point is either $0$ or $1$. Thus, to determine the number of unmatched requests (which in turn determines the instantaneous delay increment), it suffices to keep track for each point, whether there is an unmatched request at that point. This, in turn, can be done by keeping track of the parity of the number of times a request on that point has been previously matched (see the discussion following \cref{smart}). Therefore, we can simply express the current state of an MPMD-Size algorithm using an $n$-dimensional Boolean vector. Thus, the total number of MTS states is $O(2^n)$.

Finally, our reductions between MTS and MPMD-Size go in both directions. This lets us improve the lower bounds for both deterministic and randomized algorithms exponentially, and in particular, settle the deterministic competitive ratio (up to constants).

\subparagraph{MPMD-Convex on uniform metrics.}

Different from the linear delay, a convex delay function has an unbounded increase rate when it has been pending unmatched for enough long time. As a result, a reasonable algorithm should avoid this case, otherwise even a short period of waiting may lead to huge delay increase. Therefore, we introduce the concept of $T$-impatient algorithms, where a request that has been pending for time $T$ must be matched immediately once there is another available request anywhere in the metric space. In this way, we naturally divide the time axis into two parts according to whether there is a request that has been pending for longer than time $T$ at the moment. If there is no request that has been pending for longer than time $T$, then the delay increase rate of every request is upper and lower bounded by two positive constants (to guarantee the positive lower bound, the property that $f'(0)>0$ is necessary) and we can use the linear delay to approximate it up to a constant factor. As for the other case, based on the $T$-impatient property of the algorithm, the only reason that a request has been pending that long is because there is neither another pending request nor the arrival of a new request to be matched with, which means that there must be a request continuously pending throughout this time period under the optimal solution as well, which provides a lower bound of the optimal cost. Luckily, for all monotone polynomials, such a lower bound is within a constant factor of the delay incurred under the algorithm.

Therefore, to design an algorithm for monotone polynomial delays such that $f'(0)>0$ on uniform metrics, it can be reduced to designing a $T$-impatient algorithm whose competitive ratio is a constant for linear delay on uniform metrics. We show that the seminal algorithm for MPMD on tree metrics by Azar et al.~\cite{DBLP:conf/soda/AzarCK17} does not satisfy the $T$-impatient property while being implemented on uniform metrics, and thus propose a new algorithm and prove it has the desired property.

\subsection{Related Works}

\subparagraph{Linear delay.}When MPMD was first introduced by Emek et al.~\cite{DBLP:conf/stoc/EmekKW16}, the delay function they studied is the uniform linear delay. They proposed a randomized algorithm with competitive ratio $O(\log^2 n+\log \Delta)$, where $n$ is the number of points in the metric space and $\Delta$ is its aspect ratio. Azar et al.~\cite{DBLP:conf/soda/AzarCK17} improved the upper bound of the randomized competitive ratio for linear delays to $O(\log n)$ by designing an algorithm based on the randomized tree embedding technique into weighted HSTs from~\cite{DBLP:journals/jacm/BansalBMN15}. The current best lower bound for the uniform linear delay function is $\Omega(\log n/\log\log n)$ of any randomized algorithm from~\cite{DBLP:conf/approx/AshlagiACCGKMWW17}, which shows a gap between known upper and lower bounds even for the most classical delay setting.

The above results that express competitive ratios in terms of $n$ and $\Delta$ usually require the algorithm to know the metric space beforehand, and such competitive ratios are independent of the length of the request sequence. However, there are also results where the competitive ratio is expressed in terms of $m$, the number of requests in the input sequence, while the algorithm does not need to know the metric space prior to its execution: the first such result was proposed by Bienkowski et al.~\cite{DBLP:conf/waoa/BienkowskiKS17} and their deterministic algorithm is $O(m^{2.46})$-competitive. Bienkowski et al.~\cite{DBLP:conf/waoa/BienkowskiKLS18} later improved the result to $O(m)$ by an adaptation of the deterministic primal-dual moat growing technique, and independently and  concurrently, Azar and Fanani~\cite{DBLP:journals/mst/AzarF20} devised a sublinear $O(m^{0.59})$-competitive deterministic algorithm. So far, the best result that expresses the competitive ratio as a function of $m$ is an $O(\log^5 m)$-competitive deterministic algorithm by Dufay and Wattenhofer~\cite{DBLP:conf/soda/DufayW26}.

\subparagraph{Concave delay.}Azar et al.~\cite{DBLP:conf/soda/AzarRV21} considered the MPMD problem with uniform concave delays. They designed an $O(1)$-competitive deterministic algorithm for the single point and an $O(\log n)$-competitive randomized algorithm for general $n$-point metric spaces. Deryckere and Umboh~\cite{DBLP:conf/approx/DeryckereU23} generalized the primal-dual framework from~\cite{DBLP:conf/waoa/BienkowskiKLS18} and proposed an $O(m)$-competitive deterministic algorithm for uniform concave delays.

\subparagraph{Convex delay.} Liu et al.~\cite{DBLP:conf/isaac/LiuPWW18} considered the MPMD problem with uniform convex delays on $n$-point uniform metric spaces. They studied monomial functions $f(t)=t^{\alpha}$ such that $\alpha>1$ (where $t$ is the time difference between arrival and being matched for each request).
They provided a deterministic $O(n)$-competitive algorithm and showed a matching deterministic lower bound of $\Omega(n)$ for all uniform non-decreasing unbounded continuous delay functions such that $f(0)=f'(0)=0$. This result demonstrates a gap between the solutions for the cases with linear and convex delays, as the algorithm  proposed by Azar et al.~\cite{DBLP:conf/soda/AzarCK17} achieves constant competitive ratio on uniform metrics because uniform metrics can be seen as trees of height $1$. They also provided a randomized $\Omega(\log n)$ lower bound for delay functions in the same class. There is no result for convex delays beyond uniform metrics yet.

\subparagraph{Set and size-based delays.}In all of the above studies, there is one uniform delay function $f(t)$ known to the algorithm in advance that applies to every request from the input, suggesting how its delay cost increases as time goes by, and the total delay cost is the sum of delay costs incurred by each request. Deryckere and Umboh~\cite{DBLP:conf/approx/DeryckereU23} studied another model of delay function that is not accumulated request-wise but step-wise: at every time step $t$, the instantaneous delay cost is incurred as a function of the set of unmatched requests during this step. This set-delay model studies MPMD problem in the least restrictive way, as there is essentially no restriction on the structure of the delay functions and the setting is non-clairvoyant, i.e. the algorithm does not know future delays. Similar non-clairvoyant delay settings have been studied in other online problems with delays, e.g., TCP Acknowledgment Problem \cite{bhore2026onlinetcpacknowledgmentgeneral} and Subadditive Joint Replenishment Problem \cite{LeUX23,DBLP:conf/approx/Ezra0PRU24}. Unfortunately, under the most general set-delay setting, every deterministic algorithm has competitive ratio of at least $\Omega(\Delta)$, where $\Delta$ is the aspect ratio of the metric space. Instead, Deryckere and Umboh further restricted to the case of size-based delays, which have natural applications in practical settings with service-level agreements such as cloud
computing. They showed that there is an $O(2^{m})$-competitive deterministic algorithm and an $O(m^{4})$-competitive randomized algorithm by a reduction to the classical MTS problem introduced by Borodin et al.~\cite{DBLP:conf/stoc/BorodinLS87} and applying the best upper bounds for both deterministic and randomized algorithms by Borodin et al. and Bubeck et al.~\cite{DBLP:journals/siamcomp/BubeckCLL21} respectively. They also provided lower bounds in terms of $n$, which is $\Omega(n)$ for deterministic algorithms and $\Omega(\log n)$ for randomized algorithms. These upper and lower bounds show significant gaps between $m$ and $n$ in the scenarios where a request sequence of unbounded length is generated on a given finite metric.

\subparagraph{Paper organization.} We start by introducing the formal definition of MTS in \cref{MTS}. In \cref{MPMD-Size and MTS}, we show the reductions between MPMD-Size and MTS on both directions and hence complete the proof of Theorems \ref{det} and \ref{ran}. We study MPMD-Convex on uniform metrics and prove \cref{convex intro} in \cref{MPMD-convex}. We conclude by some discussion on the results as well as open problems. 

\section{Preliminaries}\label{MTS}

In this section, we introduce the definition and properties of the Metrical Task System problem, which will be used in \cref{MPMD-Size and MTS} to prove the main results of our paper on MPMD-Size. The formal definitions and notations of MPMD-Size and MPMD-Convex, two different delay models of MPMD that we study, are in Sections \ref{MPMD-Size and MTS} and \ref{MPMD-convex} respectively.

The metrical task system (MTS) problem, introduced by Borodin et al.~\cite{DBLP:conf/stoc/BorodinLS87}, can be defined as a sequence of requests $\{\rho_{t}\}_{t\geq 1}$ on a given finite metric space $(V,d)$ and an initial state $s_{0}\in V$ of the server. Each request $\rho_{t}: V\rightarrow \mathbb{R}_{\geq 0}$ is a non-negative cost function on the state space $V$, illustrating the cost of the server to process request $\rho_{t}$ in different points from $V$. Upon the arrival of request $\rho_{t}$ at every step $t\geq 1$, the server can be moved to some state $s_{t}\in V$ to process $\rho_{t}$, and this incurs \emph{service cost} of $\rho_{t}(s_{t})$. Since MTS is an online problem, the state $s_{t}$ at every step $t$ must be output without knowing future requests. The total cost of a solution $\{s_{t}\}_{t\geq 1}$ that expresses the state $s_{t}$ where the server processes request $\rho_{t}$ at every step $t\geq 1$ consists of service cost and \emph{transition cost}, where the transition cost at every step $t$ is the distance $d(s_{t-1},s_{t})$ of moving to $s_{t}$ from its previous state $s_{t-1}$. As a result, the cost of a solution $\{s_{t}\}_{t\geq 1}$ is defined as:
\begin{equation*}
    \sum_{t\geq 1}[\rho_{t}(s_{t})+d(s_{t-1},s_{t})].
\end{equation*}

For MTS problems, we say that an online algorithm $\mathcal{A}$ is \emph{lazy} if: at every step $t$ when it moves from $s_{t-1}$ to a different state $s_{t}$ to process $\rho_{t}$, it must hold that $\rho_{t}(s_{t-1})>0$. Intuitively speaking, a lazy algorithm $\mA$ only makes a move when its current state incurs positive cost for the upcoming request while always staying still if the service cost of the new request in the current state is zero.

\begin{lemmaapprep}\label{lazy}
    Given any online algorithm $\mathcal{A}$ for MTS, there is an online mechanism that transforms $\mathcal{A}$ to a lazy algorithm without increasing the cost.
\end{lemmaapprep}

\begin{proof}
    Suppose that the state of $\mathcal{A}$ to process request $\rho_{t}$ is $s_{t}$ for every $t\geq 1$, and we use $s'_{t}$ to denote the state of a lazy algorithm $\mathcal{A}'$ to process $\rho_{t}$. We now show how to transform any given $\mathcal{A}$ to a lazy algorithm $\mathcal{A}'$. Let $s_{0}=s'_{0}$ be the same initial state of the MTS instance. At every step $t$ when a new request $\rho_{t}$ arrives, if $\rho_{t}(s'_{t-1})=0$, then $s'_{t}=s'_{t-1}$; otherwise, $s'_{t}=s_{t}$.

The laziness of $\mathcal{A}'$ is trivial. We now show that the cost of $\mathcal{A}'$ is no more than that of $\mathcal{A}$. For every request $\rho_{t}$, the service cost of $\mathcal{A}'$ during step $t$ is either zero if $\rho_{t}(s'_{t-1})=0$ and $s'_{t}=s'_{t-1}$, or $\rho_{t}(s_{t})$ if $\rho_{t}(s'_{t-1})>0$ and $s'_{t}=s_{t}$. In both cases, $\rho_{t}(s'_{t})\leq \rho_{t}(s_{t})$ holds. Let $\{t_{i}\}_{i\geq 1}$ be the sequence of all steps such that $s'_{t_{i}}\neq s'_{t_{i}-1}$, then $\mathcal{A}'$ stays in the same state from step $t_{i-1}$ to $t_{i}-1$ that incurs no transition cost for every $i$, and every time when it incurs transition cost from $s'_{t_{i-1}}$ to $s'_{t_{i}}$ ($t_{0}=0$), $s'_{t_{i-1}}=s_{t_{i-1}}$ and $s'_{t_{i}}=s_{t_{i}}$. As a result, the trajectory of $\mathcal{A}'$ is a shortcut of that of $\mathcal{A}$, and the transition cost of $\mathcal{A}'$ is no more than that of $\mathcal{A}$ because of triangle inequalities.
\end{proof}

\noindent From now on, we can assume that every MTS algorithm that we use is lazy.

Let MTS-Single denote the subclass of MTS problem where every request $\rho_{t}$ incurs positive cost at exactly one point $v_{t}\in V$, i.e. $\rho_{t}(v_{t})>0$ and $\rho_{t}(v)=0$ for every $v\in V\setminus \{v_{t}\}$. When a metric space is given, we use $c_{\m}$ to denote the minimum possible worst-case competitive ratio that can be achieved by an algorithm for MTS problem on it, and $c_{\ms}$ to denote that can be achieved for MTS-Single problem. The following result is folklore for both deterministic and randomized algorithms:

\begin{lemmaapprep}\label{single}
    $c_{\m}\leq c_{\ms}$.
\end{lemmaapprep}

\begin{proof}
Given a request $\rho_{t}$ in an MTS instance, let $V_{t}:=\{v\in V\mid \rho_{t}(v)>0\}$. For each $\rho_{t}$, we construct a sequence $\tau_t$ of $|V_{t}|$ requests corresponding to $|V_{t}|$ points in $V_{t}$, where every request has cost $\rho_{t}(v)/M$ at its corresponding $v\in V_{t}$ and $0$ everywhere else and the order of these requests in the sequence does not matter. Given any $\varepsilon >0$, we choose $M$ to be a large integer such that $\rho_{t}(v)/M\leq \varepsilon\cdot \min_{v_1, v_2 \in V, v_1\neq v_2}d(v_1, v_2)$ for every $v\in V_{t}$. At every step $t$ when a new request $\rho_{t}$ of the MTS instance is received, $\tau_t$ is repeated for $M$ times. This is an online transformation from an MTS instance to an MTS-Single instance, and it is easy to see that the optimal cost of the MTS-Single instance is no more than that of the MTS instance, since every solution of the MTS instance incurs the same cost when it is applied to the MTS-Single instance.

Now suppose that we have an online algorithm $\mathcal{A}_{\ms}$ for MTS-Single. For the $M$ repetitions of $\tau_t$ in the MTS-Single instance constructed from $\rho_{t}$ in the original MTS instance, let $S_{t}$ be the set of points that $\mathcal{A}_{\ms}$ visits in this process and $v_{t}$ be the point that $\mathcal{A}_{\ms}$ stays when this process ends. The algorithm $\mathcal{A}_{\m}$ for MTS moves from $v_{t-1}$ to a point $v'\in \arg \min_{v\in S_{t}}\rho_{t}(v)$ to process $\rho_{t}$, and then moves to $v_{t}$. The transition cost of $\mathcal{A}_{\m}$ during step $t$ is no more than that of $\mathcal{A}_{\ms}$ during repetitions of $\tau_t$, as the trajectory of $\mathcal{A}_{\m}$ can be seen as a shortcut as that of $\mathcal{A}_{\ms}$. On the other hand, during each repetition of $\tau_t$, $\mathcal{A}_{\ms}$ incurs either service cost of at least $\rho_t(v')/M$ if it stays at the same point, or transition cost between different points of at least $\rho_{t}(v')/(\varepsilon \cdot M)$ according to how we choose the value of $M$. If we use $\cost_{1}(\mathcal{A})$ and $\cost_{2}(\mathcal{A})$ to denote the transition and service cost of an algorithm $\mathcal{A}$ respectively, then:
\begin{align*}
    \cost_{1}(\mathcal{A}_{\m})&\leq \cost_{1}(\mathcal{A}_{\ms});\\
    \cost_{2}(\mathcal{A}_{\m})&\leq \cost_{2}(\mathcal{A}_{\ms})+\varepsilon\cdot \cost_{1}(\mathcal{A}_{\ms}).
\end{align*}

As a result, the total cost of $\mathcal{A}_{\m}$ is no more than $(1+\varepsilon)$ times that of $\mathcal{A}_{\ms}$. $c_{\m}\leq c_{\ms}$ follows by setting $\varepsilon$ to an arbitrarily small amount.
\end{proof}

\noindent For completeness, we include the proofs of Lemmas~\ref{lazy} and~\ref{single} in the Appendix. They will be used in \cref{MPMD-Size and MTS} to establish the reductions between MPMD-Size and MTS.

\section{Reduction between MPMD-Size and MTS}\label{MPMD-Size and MTS}
In this section, we first formally define the MPMD problem with size-based delays (MPMD-Size), and then describe the reduction between it and MTS.
\subsection{Definition and Property of MPMD-Size}\label{section 3.1}

 The MPMD problem is defined on a metric space $(V,d)$, where $V$ is the set of points and $d:V\times V\rightarrow \mathbb{R}_{\geq 0}$ is the distance function. An online input consists of a sequence of requests that arrive over time. Every request $r$ is defined by two characteristics: its arrival time $a(r)$ and location $\ell(r)\in V$. We view the problem in the discrete way that any request can only arrive at a time step $t\in\mathbb{N}_{>0}$, yet it is possible that multiple requests arrive at the same time step $t$. This model captures the original definition by Emek et al.~\cite{DBLP:conf/stoc/EmekKW16} when we divide the continuous time axis into discrete steps where a step covers an interval of infinitesimal length.

Suppose $Q_t$ is the set of requests that have arrived up to and including time step $t$. At every step, the instantaneous delay function $f_{t}:2^{Q_t}\rightarrow \mathbb{R}_{\geq 0}$ specifies the increment of delay cost at this step for every possible $P_t\subseteq Q_t$, where $P_{t}$ is the set of requests that are pending during time step $t$. At every step $t$, the algorithm first sees the newly-arriving requests (if any) and the instantaneous delay function $f_{t}$, then makes a decision on matching some requests. A match between two requests $r_{1}$ and $r_{2}$ incurs a connection cost of $d(\ell(r_{1}), \ell(r_{2}))$. This process results in an updated $P_{t}$, and the instantaneous delay $f_{t}(P_{t})$ is then incurred after the matching performed at step $t$. Observe that the definition is consistent with metrical task systems (defined in \cref{MTS}) in the sense that at every step, the algorithm first sees the new input, then decides which state to move to (possibly the same state), and incurs the cost of the resulting state. The goal is to minimize the sum of its connection cost and delay cost incurred among all steps.

In the MPMD-Size problem, the instantaneous delay function incurred at time step $t$ can be expressed as $f_{t}(|P_{t}|)$, where $f_{t}:\mathbb{N}_{\geq0}\rightarrow \mathbb{R}_{\geq0}$ is a monotone function. We note that the well-studied linear delay cost is a special case of the size-based delay where $f_{t}(|P_{t}|)=|P_{t}|$ at every moment $t$.

We say that an online algorithm $\mathcal{A}$ is \emph{smart} if no two requests are pending at the same point during any time step. One important property of sized-based delay functions is that we can without loss of generality assume that every algorithm or solution is smart, because of the following lemma:

\begin{lemma}\label{smart}
    For size-based delays, given any online algorithm $\mathcal{A}$, there is an online mechanism that transforms $\mathcal{A}$ to a smart algorithm without increasing the cost.
\end{lemma}

\begin{proof}
    If $\mathcal{A}$ is not smart, suppose that $r_{1}$ and $r_{2}$ are the first appearance of two pending requests such that $\ell(r_{1})=\ell(r_{2})$ and $t_{0}\leq t<t_{m}$ is the maximal time interval where both $r_{1}$ and $r_{2}$ are pending. We construct an online algorithm $\mathcal{A}'$ that behaves the same as $\mathcal{A}$ before $t_{0}$ (since there is no violation of the smart property before $t_{0}$) and matches $r_{1}$ with $r_{2}$ at $t_{0}$. After that, the behavior of $\mathcal{A}'$ will depend on how $\mathcal{A}$ matches each of $r_{1}$ and $r_{2}$ after $t_{m}$, and we need to discuss different cases. At every time step $t$, we use $P_{t}$ to denote the set of requests pending at $t$ under $\mathcal{A}$ and $P'_{t}$ to denote the set of requests pending at $t$ under $\mathcal{A}'$, for the convenience of comparing the costs of $\mathcal{A}$ and $\mathcal{A}'$.

    If $\mathcal{A}$ matches $r_{1}$ with $r_{2}$ at time $t_{m}$, then all other matches made by $\mathcal{A}'$ are the same as that by $\mathcal{A}$. In this case, $P'_{t}=P_{t}$ for all $t<t_{0}$ and $t\geq t_{m}$ and $P_{t}=P'_{t}\cup\{r_{1}, r_{2}\}$ for $t_{0}\leq t <t_{m}$, and it is easy to see that the delay cost of $\mathcal{A}'$ is at most that of $\mathcal{A}$ while the connection cost of two algorithms is the same, given the monotonicity of the size-based delay cost. If $\mathcal{A}$ matches $r_{1}$ (this assumption is without loss of generality) with $r'_{1}$ at $t_{m}$ and $r_{2}$ with $r'_{2}$ at some time step $t'\geq t_{m}$, then $\mathcal{A}'$ will match $r'_{1}$ with $r'_{2}$ at $t'$ instead, and all other matches made by $\mathcal{A}$ and $\mathcal{A}'$ are the same. In this case, the connection cost of $\mathcal{A}$ minus that of $\mathcal{A}'$ is $d(\ell(r_{1}), \ell(r_{1}'))+d(\ell(r_{2}), \ell(r_{2}'))-d(\ell(r_{1}'), \ell(r_{2}'))$, which is non-negative by the fact that $\ell(r_{1})=\ell(r_{2})$ and the triangle inequality. As for the delay cost incurred by pending requests, it holds that
    \[
P_{t} = \begin{cases}
     P'_{t}& \text{when } t<t_{0} \text{ or }t\geq t';\\
    P'_{t}\cup\{r_{1},r_{2}\} & \text{when } t_{0}\leq t<t_{m};\\
    P'_{t}-\{r'_{1}\}+\{r_{2}\} & \text{when } t_{m}\leq t<t'.
\end{cases}
\]

Given that $r'_{1}\in P'_{t}$ when $t_{m}\leq t<t'$, we can conclude that $|P'_{t}|\leq |P_{t}|$ for every step $t$ and as a consequence, the delay cost under $\mathcal{A}'$ is no more than that under $\mathcal{A}$. After the transformation from $\mathcal{A}$ to $\mathcal{A}'$, we obtain an online algorithm where the smart property is no longer violated at $t_{0}$ and the first appearance of possible violations is delayed to a later time step. Such online transformation process and argument can be repeated and terminated within finite steps until all the occurrences of the violations are eliminated, and finally we get a smart algorithm whose cost is no more than the original $\mathcal{A}$.
\end{proof}

\cref{smart} provides us with not only a nice property of the optimal solution, but also a mechanism of online transformations that allows us to always automatically match two requests at the same point immediately at any moment.
As a result, we do not need to explicitly describe the a match of two requests at the same point in the metric anymore. Based on this, we can assume that at every point $v\in V$ there is either one or zero pending request at every time step $t$, and there are at most $n$ requests pending in the whole metric space $V$. Moreover, at every given time step, we can refer to a match between two requests at two different points as a match between these two points without ambiguity. Throughout the paper, we may use expressions such as ``the algorithm matches point $v_{1}$ with point $v_{2}$'' and ``point $v$ has been matched for an odd number of times'', where the former refers to a match between requests at these two points $v_{1}$ and $v_{2}$, and the latter refers to the total number of times where a request at $v$ has been matched with some request at another point. As we have discussed in \cref{technique}, this allows us to encode the state of an algorithm with a Boolean vector that has an entry over every point in $V$, which builds the foundation of our reductions.

\subsection{Construction of MTS Metric Space}

Before we formally describe the reduction between two problems, we first specify the corresponding metric space $(\kv,D)$ on which the request sequence for both MTS and MTS-Single problems is based, given a metric space $(V,d)$ for an MPMD-Size instance.

\begin{definition}
    Given any finite set $V$, define $\mathcal{K}^{V}:=\{\mathbf{u}\in \{0,1\}^{V}\mid\sum_{v\in V}\mathbf{u}(v)=0\mod 2\}$, where $\mathbf{u}(v)$ is the entry of $\mathbf{u}$ on a given $v\in V$ for every $\bu\in\{0,1\}^{V}$.
\end{definition}
\noindent Observe that $|\kv| = 2^{n-1}$.

\subparagraph{Intuition for $(\kv,D)$.} To have an intuitive understanding of the structure of $(\kv,D)$ defined on a metric $(V,d)$ and its addition operation, we can give a Boolean variable to every point in $V$, and see every match between two points in $V$ as an operation that flips the Boolean variable associated with these two points simultaneously. If the default setting is that the Boolean variable associated with every point in $V$ is $0$, then $\kv$ is exactly the set of states that can be attained by applying the matching operations, where each state is a Boolean vector that specifies the value of the Boolean variable associated with every point in $V$. With this intuition, if the cost of a match is exactly the distance of the two points involved in it, then the distance between two states in $\kv$ should be the minimum cost to transform one state into another by applying matching operations.

As a result, we define the concept of the induced matching of every element in $\kv$ (\cref{inducedmatching}) and use it to express the distance $D$ on $\kv$. We will see that this definition of distance indeed satisfies the min-cost transformation property between every two states (\cref{matchingcost}).

\begin{definition}(Addition operation on $\{0,1\}^{V}$)
 Let $\oplus$ denote the XOR Boolean operation over $\{0,1\}$, i.e. $a\oplus b=0$ if $a=b$ and $a\oplus b=1$ if $a\neq b$, for $a,b\in\{0,1\}$. For every $\bu_{1},\bu_{2}\in\{0,1\}^{V}$, $\bu_{1}\oplus\bu_{2}:=(\bu_{1}(v)\oplus \bu_{2}(v))_{v\in V}$.
\end{definition}

We use $\mathbf{0}$ to denote $(0)_{v\in V}$ and $\mathbf{1}$ to denote $(1)_{v\in V}$. We first note that $\mathcal{K}^{V}$ is a closed subgroup of $\{0,1\}^{V}$ under this addition operation.

\begin{claim}
For every $\bu_{1},\bu_{2}\in\kv$, we have $\bu_{1}\oplus\bu_{2}\in\kv$.  
\end{claim}
\begin{proof}
     If $\sum_{v\in V}\bu_{1}(v)=\sum_{v\in V}\bu_{2}(v)=0 \mod 2$, then $\sum_{v\in V}(\bu_{1}\oplus \bu_{2})(v)=\sum_{v\in V}(\bu_{1}(v)\oplus \bu_{2}(v))=\sum_{v\in V}\bu_{1}(v)+\sum_{v\in V}\bu_{2}(v)=0\mod 2$.
\end{proof}

\begin{definition}(Induced matching of an element in $\mathcal{K}^{V}$)\label{inducedmatching}
    Given $\bu\in\mathcal{K}^{V}$, let $T_{\bu}=\{v\in V\mid \bu(v)=1\}$. Define the \emph{induced matching} of $\bu$ to be an arbitrary
     min-cost perfect matching of points in $T_{\bu}$ in the metric space $(V,d)$.\footnote{A perfect matching exists since $|T_{\bu}|$ is even, by the definition of $\kv$.} For every min-cost perfect matching $M$ of $T_{\bu}$, we associate the vector $\bu$ with $M$.
\end{definition}

\begin{definition}(Distance $D$ on $\kv$)\label{distance}
Given $\bu_{1}, \bu_{2}\in\kv$, define their distance $D(\bu_{1},\bu_{2})$ to be the weight of a min-cost perfect matching induced by $\bu_{1}\oplus \bu_{2}$.
\end{definition}

We defer the proof of the following claim to the Appendix.

\begin{claimapprep}
    $(\kv, D)$ is a metric space.
\end{claimapprep}

\begin{proof}
    We verify all the axioms of a metric space. Symmetry follows immediately from the commutativity of the $\oplus$ operation on $\mathcal{K}^{V}$.

        \subparagraph{Positivity.} For any $\bu_{1}$ and $\bu_{2}$ from $\mathcal{K}^{V}$, $\bu_{1}\oplus \bu_{2}\in \kv$ always induces a matching of positive edge weights. This matching has zero weight if and only if it is empty, which is equivalent to that $\bu_{1}\oplus \bu_{2}=\mathbf{0}$. By the definition of the Boolean addition operation, this only happens when $\bu_{1}=\bu_{2}$.

\subparagraph{Triangle inequality.} Let $\bu_{1}, \bu_{2}$ and $\bu_{3}$ be three elements from $\mathcal{K}^{V}$. Suppose $M_{12}$ is an induced matching of $\bu_{1}\oplus \bu_{2}$ and $M_{23}$ is an induced matching of $\bu_{2}\oplus \bu_{3}$. We consider the symmetric difference $M_{13}'=M_{12} \operatorname{\Delta} M_{23}=(M_{12}\setminus M_{23})\cup (M_{23}\setminus M_{12})$. For each point $v\in V$, let $\delta(v)$ denote the set of edges incident to $v$, then 
\begin{equation*}
    |M'_{13}\cap\delta(v)|\equiv |M_{12}\cap\delta(v)|+|M_{23}\cap\delta(v)| \mod 2
\end{equation*}

It is easy to see that for every $v\in V$, both $|M_{12}\cap\delta(v)|$ and $|M_{23}\cap\delta(v)|$ have value either $0$ or $1$, because $M_{12}$ and $M_{23}$ are both perfect matchings on some even-sized subset of $V$. As a result, the equation suggests $|M'_{13}\cap\delta(v)|=1$ if and only if $|M_{12}\cap\delta(v)|$ is not equal to $|M_{23}\cap\delta(v)|$, if and only if $(\bu_{1}\oplus \bu_{2})(v)\neq (\bu_{2}\oplus \bu_{3})(v)$. Moreover, this is equivalent to that $(\bu_{1}\oplus\bu_{3})(v)=1$, because $\bu_{1}\oplus\bu_{3}=(\bu_{1}\oplus \bu_{2})\oplus (\bu_{2}\oplus \bu_{3})$.

Because $|M_{13}'\cap\delta(v)|\leq 2$ for every $v\in V$, $M_{13}'$ is a disjoint union of paths and cycles. Let $M_{13}=\{(u,v)\mid \text{There is a path in }M'_{13}\text{ with }u,v\text{ as two endpoints}\}$, then $M_{13}$ is a perfect matching over $T_{\bu_{1}\oplus \bu_{3}}=\{v\in V\mid (\bu_{1}\oplus \bu_{3})(v)=1\}$. If we use $\omega(\cdot)$ to denote the sum of all edge weights in one set, then $\omega(M_{13})\leq \omega(M'_{13})$ by triangle inequalities, and $\omega(M'_{13})\leq \omega(M_{12})+\omega(M_{23})=D(\bu_{1},\bu_{2})+D(\bu_{2},\bu_{3})$. Since $M_{13}$ is a perfect matching over $T_{\bu_{1}\oplus \bu_{3}}$, $D(\bu_{1},\bu_{3})\leq \omega(M_{13})$. As a result of all these inequalities, $D(\bu_{1},\bu_{3})\leq D(\bu_{1},\bu_{2})+D(\bu_{2},\bu_{3}) $.

Therefore, all axioms that define a metric are verified and $(\kv,D)$ is a metric space.
\end{proof}

\subsection{From MPMD-Size to MTS}

Given any instance of MPMD-Size on the metric space $(V,d)$, we first describe the reduction from it to an instance of MTS on the corresponding $(\mathcal{K}^{V},D)$, as well as how to translate an online algorithm for MTS to one for the original MPMD-Size.

We introduce some notations to express the state of an MPMD-Size instance under the execution of a given algorithm at step $t$, including: the total arrivals of requests at different points, and the effect of all matches made by the algorithm so far.

\begin{definition}\label{RM notation}
 Given an algorithm for MPMD-Size on the metric space $(V,d)$, for every step $t\geq 1$, let $R_{t}\in\{0,1\}^V$ be the Boolean vector such that $R_{t}(v)=1$ if and only if the parity of the number of requests that have arrived at $v\in V$ until (including) step $t$ is odd, $M_{t}\in\kv$ be the Boolean vector such that $M_{t}(v)=1$ if and only if $v$ has been matched for an odd number of times until (including) step $t$. The initial setting ($t=0$) is that $R_{0}=M_0=\mathbf{0}$.
\end{definition}

In MPMD-Size, at every step $t$, new requests arrive and the instantaneous delay function $f_{t}$ is revealed at the same time, and an online algorithm needs to make decisions to match according to this information. At each step $t$, we need to handle two situations: when new requests have arrived at $t$ and before the algorithm has made new decisions yet, and after the algorithm has made decisions at time $t$. The former determines the matchings that the algorithm can make, and the latter determines the instantaneous cost incurred at the end of step $t$. With notations introduced in \cref{RM notation}, it is easy to see that the former state can be expressed as $R_{t}\oplus M_{t-1}$ and the latter state can be expressed as $R_{t}\oplus M_{t}$. With the smart algorithm property and the mechanism shown in \cref{smart}, it is not hard to see that we can assume that requests at the same point cancel out each other automatically for state $R_{t}\oplus M_{t-1}$ as well. We summarize the above discussion in the following observation that expresses the number of request at every point with a Boolean vector:

\begin{observation}\label{matchandarrival}
    Upon the arrivals of new requests at time step $t$ (before new matches are made by the algorithm at step $t$), the number of pending requests at every point $v$ is equal to $M_{t-1}(v)\oplus R_{t}(v)$; once the algorithm has made decisions at step $t$, the number of pending requests at $v$ is equal to $M_{t}(v)\oplus R_{t}(v)$ when the delay cost is incurred at step $t$.
\end{observation}

If $(M_{t-1}\oplus R_t)(v)=0$, then there is no request pending at $v$ after new arrivals at $t$, and of course there cannot be any request at $v$ after new matches are performed at $t$ as well, thus leading to the following observation:
\begin{observation}\label{no new request}
    $(R_t\oplus M_{t-1})(v)\leq (R_t\oplus M_t)(v)$ for every $v\in V$ and time step $t$.
\end{observation}

For every $\bu\in \{0,1\}^{V}$, we use $|\bu|$ to denote the number of $1$'s among all entries of $\bu$, i.e. $|\bu|=\sum_{v\in V}\bu(v)$. During every step $t$ when the delay cost is incurred, $R_{t}\oplus M_{t}$ is the Boolean vector that denotes the state of the input under the algorithm as per \cref{matchandarrival}. Based on these notations and the property of size-based delays, the delay cost incurred during step $t$ can be expressed as $f_{t}(|R_{t}\oplus M_{t}|)$.

\subparagraph{Construction of MTS instance.} Now we are ready to define the MTS instance that we reduce MPMD-Size to. We set the server's initial position of the MTS instance to be $C_{0}=\mathbf{0}$. Upon every time step $t\geq 1$, we assign a request vector $(\rho_{t}({\bu}))_{\bu\in\mathcal{K}^{V}}$ to the MTS instance according to $R_{t}$, the Boolean vector that we have defined to denote the parity of the number of requests that have arrived at each point in $V$ until (and including) $t$, and $f_{t}$, the instantaneous delay cost function of the MPMD-Size instance incurred at step $t$. More specifically, the cost to process $\rho_{t}$ at $\bu$ is $\rho_{t}({\bu})=f_{t}(|R_{t}\oplus \bu|)$ for every $\bu\in\kv$.

 Let $\OPT_{\m}$  denote the value of the optimal solution of the MTS instance and $\OPT_{\mathrm{size}}$ denote that of the MPMD-Size instance.
\begin{lemma}\label{corres}
$\OPT_{\mathrm{size}}\geq \OPT_{\m}$.
\end{lemma}

\begin{proof}
    Given any solution $\SOL_{\mathrm{size}}$ of MPMD-Size, we convert it to a solution $\SOL_{\m}$ of MTS without increasing the cost. Upon every step $t$, let $M_{t}$ be the sum of Boolean vectors of matches made by $\SOL_{\mathrm{size}}$ so far, then $M_{t}$ is an element of $\mathcal{K}^{V}$, and $\SOL_{\m}$ moves from state $M_{t-1}$ to $M_{t}$ (this is feasible at step $1$ since at the beginning $M_{0}=\mathbf{0}$ and the initial configuration of the MTS instance is given as $C_{0}=\mathbf{0}$ as well).

    We now compare the cost of two solutions. At every step $t$, the number of pending requests in the MPMD-Size solution is $|R_{t}\oplus M_{t}| $. As a result, the delay cost incurred during step $t$ is $f_{t}(|R_{t}\oplus M_{t}|)$, which is equal to $\rho_{t}({M_{t}})$, the cost to process the task $\rho_{t}$ at $M_{t}$, based on our setting of $(\rho_{t}({\bu}))_{\bu\in\mathcal{K}^{V}}$. By Lemma \ref{matchingcost}, the connection cost of $\SOL_{\mathrm{size}}$ at step $t$ is at least the distance $D(M_{t-1}, M_{t})$, which equals  the cost of MTS to move from $M_{t-1}$ to $M_{t}$.
    \end{proof}

\subparagraph{Translating MTS algorithm to MPMD-Size.}
We now describe how to transform an algorithm $\mathcal{A}_{\m}$ for the MTS instance to one $\mathcal{A}_{\mathrm{size}}$ for the MPMD-Size instance in an online way. As we see from the proof of \cref{corres}, there is naturally a one-to-one correspondence between the current MTS state and the current total matching Boolean vector $M_t$ of MPMD-Size, which can be seen as the state of an MPMD-Size algorithm. Therefore, given an online algorithm $\mathcal{A}_{\m}$ for MTS problem on $(\kv, D)$ that processes request $\rho_{t}$ at state $C_{t}$, intuitively, we want to transit MPMD-Size algorithm to state $M_{t}=C_t$ as well. However, such a transition is not always feasible, as there is no guarantee that there is an available request at every point that the algorithm needs to match to transit to $C_{t}$. If this happens, then we want to update $M_{t-1}$ to $M_{t}$ that is the closest possible to $C_{t}$ among all feasible matchings that can be made: we check every pair of match in the min-cost perfect matching induced by $M_{t-1}\oplus C_{t}$, and match every feasible pair of requests. See the description of Algorithm~\ref{transformation} for more details.

\begin{algorithm}[hbt!]
\caption{The translation algorithm of an online algorithm $\mathcal{A}_{\m}$ for MTS on $(\mathcal{K}^{V}, D)$ to $\mathcal{A}_{\mathrm{size}}$ for MPMD-Size on $(V,d)$}\label{transformation}
\begin{algorithmic}
\Require $M_0=\mathbf{0}$ for MPMD-Size, and the initial configuration of MTS is $C_{0}=\mathbf{0}$.
\Ensure
\State Given $R_{t}$, the arrivals of requests so far in the MPMD-Size instance and $f_{t}$ the instantaneous size-based delay cost, task vector $\rho_{t}$ of MTS is constructed such that $\rho_t(\bu)=f_t(|R_t\oplus \bu|)$ for every $\bu\in\kv$, and we are given $\mathcal{A}_{\m}$'s state $C_{t}$ at step $t$.

Upon arrivals of new requests at $t$, the Boolean vector of current available requests in $V$ is $R_{t}\oplus M_{t-1}$ (Observation \ref{matchandarrival}).

Calculate $M_{t-1}\oplus C_{t}$, the Boolean vector of the matching that $\mathcal{A}_{\mathrm{size}}$ intends to make to transit to state $M_{t}$. For every edge $(u,v)$ in the min-cost perfect matching induced by $M_{t-1}\oplus C_{t}$, if $(M_{t-1}\oplus R_t)(u)=(M_{t-1}\oplus R_t)(v)=1$, then there is a pending request at both $u$ and $v$ after new arrivals at step $t$, and we match $u$ and $v$ in this case. 

After all matches have been performed, the new vector $M_t$ is updated accordingly.

\end{algorithmic}
\end{algorithm}

\begin{claim}\label{intended change}
    For every $v\in V$, if $M_{t}(v)\neq C_{t}(v)$, then $(M_{t-1}\oplus C_{t})(v)=1$.
\end{claim}
\begin{proof}
    We prove the contrapositive. If $(M_{t-1}\oplus C_{t})(v)=0$, then $\mathcal{A}_{\m}$ does not make any match that involves $v$ at step $t$, and so $M_{t}(v)=M_{t-1}(v)=C_{t}(v)=0$.
\end{proof}

The following claim guarantees the local optimality of $\mathcal{A}_{\m}$ as per the transition from $M_{t-1}$ to $M_t$ by matches at step $t$:
\begin{claim}\label{local optimality}
    The connection cost of $\mathcal{A}_{\m}$ at step $t$ is equal to $D(M_{t-1}, M_t)$.
\end{claim}
\begin{proof}
Let $J_1$ be the set of edges matched by $\mathcal{A}_{\m}$ at step $t$ and $J_2$ be the set of edges in the min-cost perfect matching induced by $M_{t-1}\oplus C_{t}$, then $J_1\subseteq J_2$ by the mechanism of $\mathcal{A}_{\m}$. If $J'_{1}$ is a another perfect matching on $T_{M_{t-1}\oplus M_{t}}=\{v\in V\mid (M_{t-1}\oplus M_t)(v)=1\}$ such that the cost of $J'_1$ is smaller than $J_1$, then $(J_2\setminus J_1)\cup J'_1$ is a perfect matching on $T_{M_{t-1}\oplus C_{t}}=\{v\in V\mid (M_{t-1}\oplus C_t)(v)=1\}$ whose cost is smaller than $J_2$, a contradiction to the assumption that $J_2$ is a min-cost perfect matching!
\end{proof}

We prove that the cost of $\mathcal{A}_{\mathrm{size}}$ is no more than that of $\mathcal{A}_{\m}$ by the following lemmas. First, we compare the instantaneous delay cost of $\mathcal{A}_{\mathrm{size}}$ incurred at step $t$ with the cost of processing task $\rho_{t}$ by $\mathcal{A}_{\m}$, and this is where the size-based property of the delay function $f_{t}$ plays a vital role.

\begin{lemma}\label{delay}
    The delay cost at step $t$ by $\mathcal{A}_{\mathrm{size}}$ is no more than the service cost to process request $\rho_{t}$ by $\mathcal{A}_{\m}$, given that $f_{t}$ is a size-based delay.
\end{lemma}

\begin{proof}
    The delay cost of $\mathcal{A}_{\mathrm{size}}$ at step $t$ is $f_{t}(|R_{t}\oplus M_{t}|)$ and the service cost to process request $\rho_{t}$ by $\mathcal{A}_{\m}$ is $f_{t}(|R_{t}\oplus C_{t}|)$. By the monotonicity of the size-based delay function, it suffices to show that $|R_{t}\oplus M_{t}|\leq |R_{t}\oplus C_{t}|$. We do this by constructing an injective map $h_{t}: P'_{t}\setminus P_t\rightarrow P_{t}\setminus P'_t$, where
 \begin{align*}
          P'_{t}=\{v\in V\mid (R_{t}\oplus M_{t})(v)=1 \};\qquad
          P_{t}=\{v\in V\mid (R_{t}\oplus C_{t})(v)=1\}.
 \end{align*}

For every $v\in P'_{t}\setminus P_t$, it holds that $M_{t}(v)\neq C_{t}(v)$ according to the definition of $P'_t$ and $P_t$. By \cref{intended change}, it must be that $(M_{t-1}\oplus C_{t})(v)=1$. Let $u$ be the request that is matched with $u$ in the min-cost perfect matching induced by $M_{t-1}\oplus C_{t}$. We now prove that $u\in P_t\setminus P'_t$, and hence we can define $h_t(v)=u$.

By \cref{no new request}, $(R_t\oplus M_{t-1})(v)=1$ for every $v\in P'_t$. By the construction of $\mathcal{A}_{\mathrm{size}}$, if $(R_t\oplus M_{t-1})(v)=1$ and $M_{t}(v)\neq C_{t}(v)$ both hold, then it must be that $(R_{t}\oplus M_{t-1})(u)=0$ making $v$ unmatched at step $t$. In this case, it holds that $(R_t\oplus M_{t})(u)\leq (R_t\oplus M_{t-1})(u)=0$ according to \cref{no new request}, which implies that $u
\notin P'_t$ by definition. On the other hand, $(M_{t-1}\oplus C_t)(u)=1$ yet $u$ is unmatched at step $t$, which implies that $(R_t\oplus C_t)(u)=(R_t\oplus M_{t-1})(u)\oplus 1=1$, and so $u\in P_t$. As a result, $u\in P_t\setminus P'_t$ and $h_t$ is well-defined on $P'_t\setminus P_t$.

We now prove that $h_{t}$ is injective. For any $v_1\neq v_2$ from $P'_t\setminus P_t$, $h_t(v_1)$ is matched with $v_1$ and $h_t(v_2)$ is matched with $v_2$ in the min-cost perfect matching induced by $M_{t-1}\oplus C_t$. As a result, $h_t(v_1)\neq h_t(v_2)$ by the matching property, and $h$ is an injection.
\end{proof}

The next lemma compares the connection cost of $\mathcal{A}_{\mathrm{size}}$ (which is $D(M_{t-1}, M_t)$ at every step $t$ according to \cref{local optimality}) with the transition cost of $\mathcal{A}_{\m}$.
\begin{lemma}\label{transition}
    $\sum_{t=1}^{T}D(M_{t-1}, M_{t})\leq \sum_{t=1}^{T}D(C_{t-1}, C_{t})$, where $T$ is the final step of both of the corresponding MTS and MPMD-Size instances.
\end{lemma}

\begin{proof}
    We define a potential function $\phi_{t}=D(M_{t}, C_{t})$ for each step $t$. It is easy to see that $\phi_{0}=0$ and $\phi_{t}\geq 0$ for every $t$. We prove the telescopic inequalities
    \begin{equation*}
        D(M_{t-1}, M_{t})\leq D(C_{t-1}, C_{t})-(\phi_{t}-\phi_{t-1})
    \end{equation*}
    for every step $t\geq 1$ to draw the conclusion.

    First, by the triangle inequality on the metric space $(\mathcal{K}^{V}, D)$, it holds that
    \begin{equation*}
        D(M_{t-1}, C_{t})\leq D(M_{t-1}, C_{t-1})+D(C_{t-1}, C_{t})=D(C_{t-1}, C_{t})+\phi_{t-1}.
    \end{equation*}
    
    We now show that $D(M_{t-1}, M_{t})+\phi_{t}=D(M_{t-1}, M_{t})+D(C_{t}, M_{t})\leq D(M_{t-1}, C_{t})$, then we get the desired inequality by summing up these two.

    Define
    \begin{align*}
          T_{1}=\{v\in V\mid (M_{t-1}\oplus M_{t})(v)=1\};
          \qquad T_{2}=\{v\in V\mid (M_{t}\oplus C_{t})(v)=1\}.
 \end{align*}
We first observe that $T_{M_{t-1}\oplus C_{t}}=\{v\in V\mid (M_{t-1}\oplus C_{t})(v)=1\}=T_{1}\Delta T_{2}$, according to the definition of $\oplus$ operation of Boolean vectors. By \cref{intended change}, for every $v\in T_2$, it must hold that $v\in T_{M_{t-1}\oplus C_{t}} $ and hence $v\notin T_1$, which implies that $T_1\cap T_2=\varnothing$. As a result, we can express $T_{M_{t-1}\oplus C_{t}}=T_1\sqcup T_2$\footnote{$T_1\sqcup T_2$ is the disjoint union of $T_1$ and $T_2$.} . Moreover, for each pair $(u,v)$ in the min-cost perfect matching induced by $M_{t-1}\oplus C_{t}$, either $u,v$ are both in $T_{1}$, or $u,v$ are both in $T_{2}$, according to the mechanism of $\mathcal{A}_{\mathrm{size}}$ to update $M_{t}$ at every step. Therefore, we can divide all matches in the min-cost perfect matching induced by $M_{t-1}\oplus C_{t}$ into two parts: $J_{1}$ that forms a matching on $T_{1}$, and $J_{2}$ that forms a matching on $T_{2}$. Using $w(\cdot)$ to denote the weight of a matching, we have $w(J_{1})+w(J_{2})=D(M_{t-1}, C_{t})$ since $J_{1}\cup J_{2}$ forms the min-cost perfect matching induced by $M_{t-1}\oplus C_{t}$. Moreover, $D(M_{t-1}, M_{t})\leq \omega(J_{1})$ and $D(M_{t}, C_{t})\leq\omega (J_{2})$ by the definition of distance $D$. This proves the desired inequality.
\end{proof}

Combining \cref{delay,transition}, we conclude that the total cost of $\mathcal{A}_{\mathrm{size}}$ on the MPMD-Size instance on $(V,d)$ is no more than the total cost of $\mathcal{A}_{\m}$ on the MTS instance on $(\mathcal{K}^{V},D)$, yielding the following theorem of the reduction.

\begin{theorem}
Let $(V,d)$ be a metric, $c_{\m}$ be the competitive ratio of MTS on $(\kv,D)$ and $c_{\mathrm{size}}$ be that of MPMD-Size on $(V,d)$. Then, $c_{\mathrm{size}}\leq c_{\m}$.
\end{theorem}

Applying the upper bound of $O(N)$ for deterministic algorithms from \cite{DBLP:conf/stoc/BorodinLS87} and that of $O(\log ^{2}N)$ for randomized algorithms from \cite{DBLP:journals/siamcomp/BubeckCLL21} of MTS on every $N$-point metric gives the upper bounds in Theorems~\ref{det} and~\ref{ran}.

\begin{corollary}
There is an $O(2^{n})$-competitive deterministic algorithm and an $O(n^2)$-competitive randomized algorithm for MPMD-Size on every $n$-point metric.  
\end{corollary}

\subsection{From MTS-Single to MPMD-Size}

In this part, we establish the reduction from any given instance of MTS-Single on $(\mathcal{K}^{V},D)$ to one of MPMD-Size on $(V,d)$, where $V$ is a given metric space and $|V|=n$. 

\subparagraph{MPMD-Size delay function.} We first state the structure of the size-based delay function $f_{t}$ of the MPMD-Size instance for every step $t\geq 1$:
\[
f_{t}(N) = \begin{cases}
    0 & \text{if } N<n ;\\
    \varepsilon_t & \text{if } N\geq n.
\end{cases}
\]
Here $\varepsilon_t$ denotes a positive amount in our construction that will depend on the request of the MTS-Single instance at step $t$. More specifically, at every step $t$ there is exactly one point that is ``hit'' with a positive cost by the request in MTS-Single, and we use $\varepsilon_t$ to denote the amount of this positive cost.

\subparagraph{Requests of MPMD-Size instance.} We assume that the initial setting of the MTS server is $C_0=\mathbf{0}$, otherwise we set $R_t\oplus\bu_t=C_0\oplus \mathbf{1}$ and show that the property maintained in the proof of \cref{MSSIZE} is $M_t\oplus C_t=C_0$ for every step $t$ instead (it is intuitively like ``shifting" both $R_t$ and $M_t\oplus C_t$ by $C_0$). At every step $t$ when the new request $\rho_{t}$ arrives in the MTS-Single instance on $(\kv,D)$, suppose that $\rho_{t}(\bu_{t})=\varepsilon_t>0$ and $\rho_{t}(\bu)=0$ for every $\bu\in \kv\setminus \{\bu_{t}\}$, then we set $R_t=\bu_t\oplus \mathbf{1}$ for the MPMD-Size instance. To have a more intuitive understanding of how new requests arrive in the metric space at step $t$, we use the difference between $R_{t-1}$ and $R_{t}$, i.e. $R_{t-1}\oplus R_{t}$, to denote the new arrivals at step $t$, such that $(R_{t-1}\oplus R_{t})(v)=1$ if and only if a new request arrives at step $t$ at point $v$, for every $v\in V$. Moreover, we can express this amount as $R_{t-1}\oplus R_{t}=R_{t-1}\oplus \bu_t\oplus \mathbf{1}$, which can be computed from $R_{t-1}$ that has been well-defined from previous steps, and $\bu_t$ defined by the new request of MTS-Single instance at step $t$. As a result, the MPMD-Size instance that we reduce MTS-Single to is well-defined.

The following lemma provides a lower bound of a solution's connection cost at step $t$, given $M_{t-1}$ and $M_{t}$, the two consecutive states of matches that have been made so far.

\begin{lemma}\label{matchingcost}
    The cost of matches made at step $t$ is at least $D(M_{t-1}, M_{t})$.
\end{lemma}

\begin{proof}
Let $M(v)$ denote the number of times that the algorithm matches $v$ with another point at step $t$, then $M(v)$ is either $0$ or $1$ based on the assumption that there is at most one available request at every $v\in V$ and so we naturally regard $M$ as a Boolean vector. For every $v\in V$, it holds that $M_{t}(v)=M_{t-1}(v)\oplus M(v)$, according to the definition of $M_{t}$ and $M_{t-1}$ and the addition operation on Boolean vectors. As a result, $M$ is a perfect matching of $T_{M_{t-1}\oplus M_{t}}=\{v\in V\mid (M_{t-1}\oplus M_{t})(v)=1\}$ (since we only consider matches between different points). The statement follows by \cref{distance}.
\end{proof}

Let $\OPT_{\ms}$ denote the value of the optimal solution of the MTS-Single instance and $\OPT_{\mathrm{size}}$ denote the that of the MPMD-Size instance.

\begin{lemma}\label{MSSIZE}
    $\OPT_{\ms}\geq \OPT_{\mathrm{size}}$. 
\end{lemma}

\begin{proof}
Given any solution $\SOL_{\ms}$ of MTS-Single, we convert it to a solution $\SOL_{\mathrm{size}}$ of MPMD-Size without increasing the cost as follows. In the solution $\SOL_{\mathrm{size}}$, a match only happens at time step $t$ when $\SOL_{\ms}$ transitions from $C_{t-1}$ to a different state $C_{t}\neq C_{t-1}$. By \cref{lazy}, we can always assume that $\bu_t=C_{t-1}$ is the only non-zero entry of the task vector $\rho_{t}$ of time step $t$. As for the correspondence between $\SOL_{\ms}$ and $\SOL_{\mathrm{size}}$, $\SOL_{\mathrm{size}}$ would make a min-cost perfect matching induced by $C_{t-1}\oplus C_{t}$ with cost equal to $D(C_{t-1}, C_{t})$. We need to show that such a matching is always feasible upon the arrival of new requests at time step $t$, i.e. that there is an unmatched request on each point $v$ with $(C_{t-1} \oplus C_t)(v)=1$. 

We claim that the value of $M_t\oplus C_t$ does not change, i.e., the equality $M_{t}=C_{t}$ is maintained for every $t$ because $M_0=C_0=\mathbf{0}$ as per the initial setting. If this holds, then at every step $t$ when $\rho_{t}$'s non-zero entry $\bu_t$ is the same as $C_{t-1}$, it holds that $\bu_t=C_{t-1}=M_{t-1}$ and hence $R_{t}\oplus M_{t-1}=R_{t}\oplus \bu_t=\mathbf{1}$ according to our construction of the MPMD-Size instance. As a result, upon the arrival of new requests at step $t$, there is a request at every point in $V$, therefore every matching is feasible at step $t$.

The claim trivially holds for $t=0$ because $M_{0}=C_{0}$ in the initial setting. Upon a new step $t$, if $\SOL_{\ms}$ does not change its state at step $t$, i.e.~$C_{t}=C_{t-1}$, then $\SOL_{\mathrm{size}}$ does not make any match at step $t$ as well, and $M_{t}=M_{t-1}=C_{t-1}=C_{t}$ is maintained.  If $\SOL_{\ms}$ moves to a different state $C_{t}\neq C_{t-1}$, then $\bu_t=C_{t-1}$ according to our previous assumption. Based on the definition of the MPMD-Size instance, $R_{t}\oplus \bu_t=\mathbf{1}$, which means that $R_{t}\oplus M_{t-1} =R_{t}\oplus C_{t-1}=R_{t}\oplus \bu_t=\mathbf{1}$ because of the inductive hypothesis that $C_{t-1}=M_{t-1}$. As a result, the matching induced by $C_{t-1}\oplus C_{t}$ is feasible when new requests arrive at step $t$, and $M_{t}=M_{t-1}\oplus (C_{t-1}\oplus C_{t})=(M_{t-1}\oplus C_{t-1})\oplus C_{t}=C_{t}$. Therefore, the statement is maintained and this concludes the proof of the claim.

    Based on the discussion above, the connection cost of $\SOL_{\mathrm{size}}$ is exactly equal to the cost of moving between different states of $\SOL_{\ms}$. Moreover, $\SOL_{\mathrm{size}}$ incurs a positive cost $\varepsilon_t$ at step $t$ if and only if $R_{t}\oplus M_{t}=\mathbf{1}$ due to the structure of the delay function $f_{t}$, and this is equivalent to $R_t\oplus C_t=\mathbf{1}$ because $M_{t}=C_{t}$. Since we construct $R_{t}$ for MPMD-Size following the rule that $R_{t}\oplus \bu_t=\mathbf{1}$ where $\bu_t$ is the only non-zero entry of request $\rho_{t}$ of MTS-Single, and $\SOL_{\ms}$ incurs a positive cost $\varepsilon_t$ at step $t$ if and only if $\bu_t=C_{t}$, we get that $\SOL_{\ms}$ incurs the cost $\varepsilon_t$ if and only if $R_{t}\oplus C_{t}=\mathbf{1}$ as well. As a result, the delay cost of $\SOL_{\mathrm{size}}$ is equal to the cost of processing tasks of $\SOL_{\ms}$, and therefore $\SOL_{\mathrm{size}}=\SOL_{\ms}$.
\end{proof}

The next step is to convert an online algorithm $\mathcal{A}_{\mathrm{size}}$ for the MPMD-Size instance to one $\mathcal{A}_{\ms}$ for the MTS-Single instance without increasing the cost. At every step $t$ after new matches have been made by $\mathcal{A}_{\mathrm{size}}$ and $M_t$ is updated, $\mathcal{A}_{\ms}$ processes request $\rho_{t}$ in state $M_{t}$.

\begin{lemma}
    The cost of $\mathcal{A}_{\ms}$ is no more than that of $\mathcal{A}_{\mathrm{size}}$.
\end{lemma}

\begin{proof}
    It is easy to see the transition cost of $\mathcal{A}_{\ms}$ is no more than the connection cost of $\mathcal{A}_{\mathrm{size}}$. This is because the matching made by $\mathcal{A}_{\mathrm{size}}$ at $t$ must be at least the value of a min-cost perfect matching induced by $M_{t}\oplus M_{t-1}$ according to Lemma~\ref{matchingcost}, which is equal to $D(M_{t-1}, M_{t})$, the cost of moving from the previous state $M_{t-1}$ to a new state $M_{t}$ at step $t$ in the MTS-Single instance.

On the other hand, the delay cost incurred for MPMD-Size during step $t$ is equal to $\varepsilon_t=\rho_{t}(\bu_t)$ (recall that $\bu_t$ is the only state in $\kv$ where $\rho_{t}$ incurs positive cost) if and only if $R_t\oplus M_t=\mathbf{1}$, which is equivalent to that $M_{t}=\bu_t$ based on our setting of the MPMD-Size instance that $R_{t}\oplus \bu_t=\mathbf{1}$. Thus, the service cost of $\mathcal{A}_{\ms}$ is equal to the delay cost of $\mathcal{A}_{\mathrm{size}}$.
\end{proof}

Now the reduction from MTS-Single to MPMD-Size has been established as the following:
\begin{theorem}
Let $(V,d)$ be a metric, $c_{\ms}$ be the competitive ratio of MTS-Single on $(\kv,D)$ and $c_{\mathrm{size}}$ be that of MPMD-Size on $(V,d)$. Then, $c_{\ms}\leq c_{\mathrm{size}}$.
\end{theorem}

Combined with \cref{single}, we can use $c_{\m}$, the competitive ratio of the general MTS problem, to provide lower bounds for the MPMD-Size problem. Therefore, we obtain the following corollary when we apply the universal lower bound of $\Omega(N)$ for deterministic algorithms from \cite{DBLP:conf/stoc/BorodinLS87} and that of $\Omega(\log N)$ for randomized algorithms from \cite{DBLP:conf/stoc/BubeckCR23} of MTS on every $N$-point metric.

\begin{corollary}
No deterministic algorithm can achieve competitive ratio better than $\Omega(2^{n})$ and no randomized algorithm can achieve competitive ratio better than $\Omega(n)$ for MPMD-Size on any $n$-point metric.    
\end{corollary}

This gives the lower bounds in Theorems~\ref{det} and~\ref{ran}.

\section{MPMD-Convex on Uniform Metrics}\label{MPMD-convex}

In this section, we study another form of delay functions that does not fall in the category of size-based delays. We use the continuous time model to study the online matching with uniform delays, where the same monotone continuous delay function $f:\mathbb{R}_{\geq 0}\rightarrow \mathbb{R}_{\geq 0}$ (we further assume that $f(0)=0$ without loss of generality) applies to every request $r$ in the following way: suppose the arrival time of 
a request $r$ is $a(r)$ and $t$ is the time when it is matched by the algorithm, then the total delay cost incurred by request $r$ is $f(t-a(r))$. In this section, we still use $\ell(r)$ to denote the location of request $r$ in the metric space, and the connection cost of matching two requests $r_{1}$ and $r_{2}$ at locations $\ell(r_{1})$ and $\ell(r_{2})$ respectively is still $d(\ell(r_{1}), \ell(r_{2}))$. The total cost of a solution consists of its total connection cost of matching request into pairs, and the sum of delay cost incurred by every request.

We study monotone convex polynomial delay functions $f$ such that $f'(0)>0$ on $n$-point uniform metrics where the distance between any two different points is $2\delta>0$. Equivalently, such a metric can be interpreted as a star metric where leaf nodes are points in the metric space, the distance between every leaf and the root is $\delta$, and the distance between two leaves is the length of the shortest path between them, which is $2\delta$. The previous study by Liu et al.~\cite{DBLP:conf/isaac/LiuPWW18} gives a tight $\Theta(n)$ deterministic bound for a class of convex delay functions (including all monotone polynomials) such that $f(0)=f'(0)=0$, and our study complements the landscape by the following result for functions that $f'(0)>0$:

\begin{theorem}
\label{thm-convex}
For any fixed convex monotone polynomial delay $f$ such that $f'(0)>0$, there is a constant-competitive deterministic algorithm for MPMD under delay function $f$ on all $n$-point uniform metrics with inter-point distance $\delta$ for constant $\delta$.
\end{theorem}

The contrast of this result compared with the $\Omega(n)$ deterministic lower bound of the case of continuous non-decreasing unbounded delays satisfying that $f(0)=f'(0)=0$ demonstrates that even the local property of the delay function at time $t=0$ significantly affects the solution and competitiveness of the online MPMD problem. Different from the linear delay case, here the distance $2\delta$ between different points in the metric space also affects the competitive ratio: our algorithm is constant competitive under a fixed function $f$ when $\delta$ is considered as a constant
(the same as the study by Liu et al.~\cite{DBLP:conf/isaac/LiuPWW18}), while the competitive ratio can also be regarded as a function of $\delta$, which provides another parameter that might be interesting to study for convex delays.

\subsection{Simpler Polynomials}
Without loss of generality, we study $f$ in the form $f(t)=\frac{1}{k}t^{k}+t$ for $k\geq 2$. To see why we can make this assumption, suppose that we have another polynomial $h(t)=\sum_{i=1}^{k}a_{i}t^i$ that is monotone on $\mathbb{R}_{\geq 0}$, such that $a_{k}> 0, a_{1}>0$. \footnote{We need $a_1> 0$ to have $h'(0)>0$.} The following claim proves that $h$ can be approximated by $f$ within constant factors. 

\begin{claim}\label{approx}
There exist constants $A,B>0$ such that $A\leq\frac{f(t)}{h(t)}\leq B$ for every $t \geq 0$.    
\end{claim}

\begin{proof}
Consider the rational function $\frac{f(t)}{h(t)}$. Since $h(0)=0$ and $h$ is monotone on $\mathbb{R}_{\geq 0}$, it holds that $h(t)>0$ for every $t>0$. As a result, $\frac{f(t)}{h(t)}$ is well-defined for every $t\in(0,+\infty)$.   

Moreover, it holds that $\lim_{t\rightarrow0^{+}}\frac{f(t)}{h(t)}=\frac{1}{a_{1}}$ according to L'Hôpital's rule and $\lim_{t\rightarrow+\infty}\frac{f(t)}{h(t)}=\frac{1}{ka_{k}}$, and both values are strictly bigger than zero based on our assumption that $a_{k}$ and $a_{1}$ are both non-zero values and $h(t)$ is monotone on $\mathbb{R}_{\geq0}$. Based on the definition of the limit, there exists $T>0$ such that $0<\frac{1}{2ka_{k}}<\frac{f(t)}{h(t)}<\frac{2}{ka_{k}}$ for every $t>T$. Because $\frac{f(t)}{h(t)}$ is a continuous positive function on the closed interval $[0,T]$ (define $\frac{f(0)}{h(0)}=\frac{1}{a_1}$), it has positive upper and lower bounds on $[0,T]$. Combined with the positive upper and lower bounds on $(T,+\infty)$, we conclude that $\frac{f(t)}{h(t)}$ has positive upper and lower bounds on $\mathbb{R}_{\geq 0}$.
\end{proof}

As a result, if there is a competitive algorithm for delay function $f$, then the same algorithm can be applied to delay function $h$ and the lost factor of the competitive ratio is within a constant factor that only depends on $A$ and $B$, independent of both $n$ and $\delta$.

If we apply this idea of comparing the costs of the same solution under different delay functions to comparing the convex delay function $f(t)=\frac{1}{k}t^{k}+t$ with the classical linear delay function $g(t)=t$ on the same request sequence, the only difference that can lead to different costs and different optimal solutions are delay functions $f$ and $g$. In the following part, we focus on the linear delay function $g$, and prove that a competitive algorithm for $g$ with certain properties also works well for $f$. Given an algorithm $\mathcal{A}$ and some delay function $h$, we use $\cost_{\mathcal{A}}(h)$ to denote the cost of $\mathcal{A}$ when the delay function is set to $h$. As the total cost consists of connection cost and delay cost, we use $\cost^{d}_{\mathcal{A}}(h)$ to denote its connection cost and $\cost^{t}_{\mathcal{A}}(h)$ to denote its delay cost. Similarly for any solution $\SOL$, we use $\SOL(h)$ to denote its total cost under $h$, as well as $\SOL^d(h)$ and $\SOL^t(h)$ to denote its connection and delay costs respectively. These are the notations that we will use in computations in the next section.  

\subsection{$T$-Impatient Algorithms}

For an online algorithm $\mathcal{A}$ and a request $r$, define $t_{\mathcal{A}}(r)$ to be the time that $r$ is matched by $\mathcal{A}$. 

\begin{definition}[$T$-impatient algorithms]
    An online algorithm $\mathcal{A}$ is \emph{$T$-impatient} for a fixed value $T>0$ if for every request $r$ such that $t_{\mathcal{A}}(r)-a(r)>T$, $r$ is the only pending request in the whole metric space at every moment throughout the whole time interval $[a(r)+T,t_{\mathcal{A}}(r))$, 
    i.e.~all previous requests have been matched and no new requests arrive during the interval.
\end{definition}

 Intuitively, a $T$-impatient algorithm will only leave a request $r$ unmatched for more than $T$ time only if there is no other request it can match $r$ to. Equivalently, if a request $r$ has been pending for time $T$, then $\mathcal{A}$ is going to match $r$ with another request immediately as long as there is such a request available. The $T$-impatient property guarantees a lower bound of the optimal solution during the period when there is a request pending for longer than $T$. On the other hand, when no request is pending for longer than $T$, the delay increment can be approximated by the linear delay up to a constant factor as long as $f'(0)>0$. These two properties guarantee the competitiveness of a $T$-impatient algorithm under convex delays.

The next lemma is the key to guaranteeing that the cost of a $T$-impatient algorithm during the period when there is a request that has been pending beyond time $T$ is within a constant factor of the cost incurred by the optimal solution during the same period.
\begin{lemma}\label{convex}
    Given $f(t)=\frac{1}{k}t^{k}+t$ and a constant $T>0$, $\frac{f(x)-f(T)}{f(x-T)}< 2^{k}(T^{k-1}+1)$ for every $x>T$.
\end{lemma}

\begin{proof}
Let $x > T > 0$. Then, $x^{k}-T^{k}>(x-T)^{k}$ and hence 
\[\frac{f(x)-f(T)}{f(x-T)}=\frac{x^{k}-T^{k} + (x - T)}{(x-T)^{k} + (x - T)} < \frac{x^{k}-T^{k}}{(x-T)^{k}}.\]
If $x>2T$, then we have
\begin{equation*}
    \frac{x^{k}-T^{k}}{(x-T)^{k}}<\frac{x^{k}-T^{k}}{(x/2)^{k}}=2^{k}-\left(\frac{2T}{x}\right)^{k}<2^{k}-1<2^{k}\left(T^{k-1}+1\right),
\end{equation*}
as desired. 
On the other hand, if $T<x\leq 2T$, then we factorize $x^k - T^k$ as follows 
\begin{align*}
 x^k-T^k=\sum_{j=0}^{k-1}(x^{k-j}T^j-x^{k-j-1}T^{j+1}) =(x-T)\left(\sum_{j=0}^{k-1}x^{k-1-j}T^j\right), 
\end{align*}
and so
\begin{align*}
    \frac{f(x)-f(T)}{f(x-T)}&=\frac{\frac{1}{k}(x^{k}-T^{k})+(x-T)}{\frac{1}{k}(x-T)^{k}+(x-T)}=
    \frac{\frac{1}{k}\sum_{j=0}^{k-1}x^{k-1-j}T^{j}+1}{\frac{1}{k}(x-T)^{k-1}+1}\\
    &<\frac{1}{k}\sum_{j=0}^{k-1}x^{k-1-j}T^{j}+1\leq \frac{1}{k}\sum_{j=0}^{k-1}(2T)^{k-1-j}T^{j}+1\\
    &\leq (2T)^{k-1}+1\leq 2^{k}(T^{k-1}+1).
\end{align*}
We conclude that the statement holds for every $x>T$.
\end{proof}
We now formally state the reduction from convex delay $f(t)=\frac{1}{k}t^k+t$ to linear delay $g(t)=t$ under any $T$-impatient algorithm.
\begin{lemma}\label{impatientreduction}
    On every $n$-point uniform metric, if a $T$-impatient algorithm $\mathcal{A}$ is $\gamma$-competitive ($\gamma\geq 1$) under delay function $g(t)=t$, then $\mathcal{A}$ is $2^{k+1}(T^{k-1}+1)\gamma$-competitive under delay function $f(t)=\frac{1}{k}t^{k}+t$.
\end{lemma}
\begin{proof}
    We divide the whole time axis into disjoint parts according to the execution of an algorithm $\mathcal{A}$:
\begin{equation*}
    \Gamma_{\text{within } T}=\{t\mid t-a(r)\leq T\text{ for every request }r\text{ pending at }t\text{, or there is no request pending at $t$}\}
\end{equation*}
and
\begin{equation*}
    \Gamma^{q}_{\text{beyond } T}=\{t\mid t-a(q)> T\text{ for request }q\text{ pending at }t\}
\end{equation*}
for every request $q$ such that $t_{\mA}(q)-a(q)>T$.

It is easy to see that for any two distinct requests $q_{1}$ and $q_{2}$ such that $t_{\mA}(q_{1})-a(q_{1})>T$ and $t_{\mA}(q_{2})-a(q_{2})>T$, $\Gamma^{q_{1}}_{\text{beyond } T}\cap\Gamma^{q_{2}}_{\text{beyond } T}=\varnothing$, otherwise the intersection of two left-open intervals must have positive length corresponding to a period where both $q_1$ and $q_2$ are unmatched, a contradiction to the $T$-impatient property. As a result, $(\cup_{q\text{: }t_{\mA}(q)-a(q)>T}\Gamma^{q}_{\text{beyond } T})\cup\Gamma_{\text{within } T}$ is indeed a disjoint and complete decomposition of the entire time axis under the execution of $\mA$.

Fix some other solution $\SOL$. We first compare the costs of $\mathcal{A}$ under $f$ and $g$ as well as the cost of $\SOL$. 
For the connection cost, it holds that
\begin{equation*}
    \cost^{d}_{\mA}(f)=\cost^{d}_{\mA}(g)
\end{equation*}
since the matched pairs of requests are the same.

We can bound the delay cost that is incurred during $\Gamma_{\text{within }T}$ based on the fact that the increasing rate of the delay of any request can be upper and lower bounded by positive constants throughout this time period. We use $\cost^t_\mathcal{A}(f)\mid_{\Gamma_{\text{within }}}$ and $\cost^t_\mathcal{A}(g)\mid_{\Gamma_{\text{within }}}$ to denote the delay cost of $\mathcal{A}$ during $\Gamma_{\text{within }T}$ under delay functions $f$ and $g$ respectively, and similarly notations for $\SOL$ and $\Gamma^{q}_{\text{beyond } T}$ work in the same way. For every request $r$, let $t_{\mA}'(r)=\min\{t_{\mA}(r),a(r)+T\}$, then
\begin{align*}
    \cost^{t}_{\mA}(f)\mid_{\Gamma_{\text{within } T}}&=\sum_{r}f(t'_{\mA}(r)-a(r))=\sum_{r}[\frac{1}{k}(t'_{\mA}(r)-a(r))^{k}+(t'_{\mA}(r)-a(r))]\\
        &=\sum_{r}(t'_{\mA}(r)-a(r))\left[\frac{1}{k}(t'_{\mA}(r)-a(r))^{k-1}+1\right]\\
        &=\sum_{r}\left[\frac{1}{k}(t'_{\mA}(r)-a(r))^{k-1}+1\right]g(t'_{\mA}(r)-a(r))\\
        &\leq\left(\frac{1}{k}T^{k-1}+1\right)\sum_{r}g(t'_{\mA}(r)-a(r))\\
        &=\left(\frac{1}{k}T^{k-1}+1\right)\cost^{t}_{\mA}(g)\mid_{\Gamma_{\text{within } T}}.
\end{align*}
The only remaining task is to consider $\cost^{t}_{\mA}(f)\mid_{\Gamma^{q}_{\text{beyond }T}}$ for each request $q$ such that $t_{\mA}(q)-a(q)>T$. By the $T$-impatient property of our algorithm $\mA$, no new request arrives during $\Gamma^{q}_{\text{beyond }T}$, and there is exactly one request pending throughout this time interval under the execution of $\mA$. As a result, the total number of requests that have appeared is an odd number, which implies that under $\SOL$, there must be at least one request pending throughout this whole time period. Suppose that $\Gamma^{q}_{\text{beyond }T}$ begins at $t_{1}(q)$ and ends at $t_{2}(q)$, then $t_{1}(q)-a(q)=T$ by  definition of $\Gamma^{q}_{\text{beyond }T}$, and so \[\cost^{t}_{\mA}(f)\mid_{\Gamma^{q}_{\text{beyond }T}}=f(t_{2}(q)-a(q))-f(t_{1}(q)-a(q))=f(T+t_{2}(q)-t_{1}(q))-f(T).\] Since there is at least one request in $\SOL$ that is continuously pending during the whole interval $[t_{1}(q), t_{2}(q))$, we have $\SOL^{t}(f)\mid_{\Gamma^{q}_{\text{beyond }T}} \geq f(t_{2}(q)-t_{1}(q))$. By \cref{convex}, we have \[f(T+t_{2}(q)-t_{1}(q))-f(T)\leq 2^{k}(T^{k-1}+1)f(t_{2}(q)-t_{1}(q)),\] and so \[\cost^{t}_{\mA}(f)\mid_{\Gamma^{q}_{\text{beyond }T}}\leq 2^{k}(T^{k-1}+1)\SOL^{t}(f)\mid_{\Gamma^{q}_{\text{beyond }T}}.\]

Combining the bounds on $\cost^{t}_{\mA}(f)\mid_{\Gamma_{\text{within } T}}$ and $\cost^{t}_{\mA}(f)\mid_{\Gamma^{q}_{\text{beyond }T}}$, we have
\begin{align*}
    \cost^{t}_{\mA}(f)&=\cost^{t}_{\mA}(f)\mid_{\Gamma_{\text{within } T}}+\sum_{q\text{: }t_{\mA}(q)-a(q)>T}\cost^{t}_{\mA}(f)\mid_{\Gamma^{q}_{\text{beyond }T}}\\
        &\leq (\frac{1}{k}T^{k-1}+1)\cost^{t}_{\mA}(g)\mid_{\Gamma_{\text{within } T}}+2^{k}(T^{k-1}+1)\sum_{q\text{: }t_{\mA}(q)-a(q)>T}\SOL^{t}(f)\mid_{\Gamma^{q}_{\text{beyond }T}}\\
        &\leq2^{k}(T^{k-1}+1)\cost^{t}_{\mA}(g)+2^{k}(T^{k-1}+1)\sum_{q\text{: }t_{\mA}(q)-a(q)>T}\SOL^{t}(f)\mid_{\Gamma^{q}_{\text{beyond }T}}.
\end{align*}
Thus,
\begin{align*}
    \cost_{\mA}(f)&=\cost^{d}_{\mA}(f)+\cost^{t}_{\mA}(f)\\
        &\leq\cost^{d}_{\mA}(g)+2^{k}(T^{k-1}+1)\cost^{t}_{\mA}(g)+2^{k}(T^{k-1}+1)\sum_{q\text{: }t_{\mA}(q)-a(q)>T}\SOL^{t}(f)\mid_{\Gamma^{q}_{\text{beyond }T}}\\
      \intertext{and by $\gamma$-competitiveness of $\mathcal{A}$ for $g$,}
        &\leq 2^{k}(T^{k-1}+1)\gamma \cdot\SOL(g)+2^{k}(T^{k-1}+1)\SOL(f)\\
        &\leq2^{k+1}(T^{k-1}+1)\gamma\cdot\SOL(f),
\end{align*}
where the last inequality holds because $\SOL(g)\leq \SOL(f)$, due to the fact that the connection costs are the same for every solution under $f$ and $g$, and the delay cost under $f$ is always at least that under $g$ since $f(t)\geq g(t)$ for every $t\geq 0$ (this is where the property that $f'(0)>0$ is vital, otherwise if $f'(0)=0$, then $\lim_{t\rightarrow}\frac{f(t)}{g(t)}=0$ and such an inequality cannot even hold for any constant factor).

As a result, given that $\mathcal{A}$ is $\gamma$-competitive under linear delay $g(t)=t$, the same algorithm is $2^{k+1}(T^{k-1}+1)\gamma$-competitive under convex delay $f(t)=\frac{1}{k}t^{k}+t$.
\end{proof}
With this reduction, now the goal is to find an $T$-impatient algorithm for some constant $T$ independent of $n$ that achieves competitive ratio $\gamma$ under delay function $g(t)=t$, where $\gamma$ is a constant independent of $n$ as well.

\begin{remark}\label{remarknew}
At this stage, it is natural to ask whether the previous algorithm that works well for the linear delay on a tree from \cite{DBLP:conf/soda/AzarCK17} (as uniform metrics are a special class of trees with height $1$ when we view the graph as a star) satisfies the $T$-impatient property. Unfortunately, this is not true and we do have to propose a new algorithm to be applied to convex delays. See Appendix for a brief discussion on how their algorithm runs on a uniform star metric, and why it does not satisfy the $T$-impatient property for any constant $T$ even when $\delta$ is a constant.

\begin{toappendix}
We prove the statement in \cref{remarknew} that the algorithm from~\cite{DBLP:conf/soda/AzarCK17} does not satisfy the desired $T$-impatient property for any constant $T>0$.

While being implemented on a uniform star metric where the distance between each leaf and the root is $\delta$, the algorithm assigns a counter $z_{v}$ to every leaf node $v\in V$ and initially the value of every counter is set to zero. The algorithm always matches two requests at the same location automatically, and there is either one or zero request pending at a given location at every moment, except the exact moment when a match is made. At any moment during the execution of the algorithm, $z_{v}$ is increasing if and only if there is one request at $v$ and $z_{v}<2\delta$. The algorithm matches two requests $r_{1}$ and $r_{2}$ such that $\ell({r_{1}})\neq \ell({r_{2}})$ when $z_{\ell({r_{1}})}=z_{\ell({r_{2}})}=2\delta$, and after this match both $z_{\ell({r_{1}})}$ and $z_{\ell({r_{2}})}$ are reset to zero.

Consider the following instance: suppose that the $n$ points in the uniform metric are $v_{1},\ldots,v_{n}$ respectively, then there is one request at $v_{1}$, one request at $v_{n}$, and two requests at each of $v_{2},\ldots,v_{n-1}$. We use $r_{i}$ and $r'_{i}$ to denote the two requests that arrive at $v_{i}$ for $2\leq i\leq n-1$, and $r_{1}$ and $r_{n}$ to denote the requests at $v_{1}$ and $v_{n}$ respectively. It holds that $a(r_{1})=0$ and $a(r_{n})=n-1$. For every $2\leq i\leq n-1$, $a(r_{i})=i-1, a(r_{i})=i-\varepsilon$ for some $\varepsilon>0$. If we run the aforementioned algorithm on this instance, then it matches $r_{i}$ with $r_{i}'$ for every $2\leq i\leq n-1$ at time $i-\varepsilon$, and matches $r_{1}$ with $r_{n}$ at time $n$. As a result, its match of request $r_{1}$ cannot satisfy the $T$-impatient property for any constant $T$ independent of $n$.
\end{toappendix}
\end{remark}

\subsection{Constructing a $T$-Impatient Algorithm}

In this section, we propose a $T$-impatient algorithm and prove that it is competitive under linear delay $g(t)=t$, and thus a competitive algorithm under $f(t)=\frac{1}{k}t^k+t$ according to \cref{impatientreduction}.

\subparagraph{High-level description.}
Recall that a uniform metric can be regarded as a star where each leaf corresponds to a point in the metric space, and the distance between each leaf and the root is $\delta$. We assign a counter to every vertex of the star graph, i.e. all leaves and the root.

During the execution of our algorithm (\cref{alg}), each request can be in one of  three states: $\LEAF$, $\ROOT$ and $\READY$. A request $r$ is in $\LEAF$ if and only if it is contributing to the increase of the leaf counter $z_{\ell(r)}$ ($\ell(r)$ is the location of $r$), in $\ROOT$ if it is neither in $\LEAF$ nor in $\READY$. If we want to view the dynamic of a request moving between states, then after a request $r$ arrives, initially it is in $\LEAF$ until $z_{\ell(r)}$ is saturated (if $z_{\ell(r)}$ is already saturated at arrival time $a(r)$, we can still say that the time length when $r$ is in $\LEAF$ is zero). Once $z_{\ell(r)}$ is saturated, $r$ transits to $\ROOT$, and will only transit from $\ROOT$ to $\READY$ if the condition of either Case 1 or Case 2 is satisfied. The state transition diagram $\LEAF\rightarrow\ROOT\rightarrow \READY$ is acyclic, which means that once a request $r$ transits to the next state, there is no way back to the previous state. Based on the fact that the counter $z_{\ell(r)}$ is saturated is equivalent to that $r$ is in $\ROOT$ or $\READY$ for a pending request $r$, the conditions to trigger a match between two requests at different points can also be expressed as: either both are in $\ROOT$, or one is in $\READY$. %
\begin{algorithm}[hbt!]
\caption{A Deterministic Algorithm for MPMD on $n$-point uniform metrics with linear delay $g(t)=t$}\label{alg}
\begin{algorithmic}
\Require For each leaf vertex $v$, $z_{v}\gets0$. For the root vertex $u$, $z_{u}\gets0$.
\Ensure
\State While there is a request $r$ pending at $\ell(r)$:
\begin{enumerate}
    \item if $z_{\ell(r)}<\delta$, then $r$ contributes to the increase of the counter $z_{\ell(r)}$ with unit rate and $r$ is in the state $\LEAF$, otherwise if $z_{\ell(r)}=\delta$ then $r$ is in $\ROOT$ until it transits to $\READY$;
    \item if $z_{\ell(r)}=\delta$ and $r$ is the only pending request, then then $r$ contributes to the increase of the root counter $z_{u}$;
    \item\label{Case 1} if the single request $r$ has contributed the total value $\delta$ to counter $z_{u}$, then $r$ transits from $\ROOT$ to $\READY$ (Case 1);
    \item\label{case2} if the total amount of time when $r$ is in $\ROOT$ while not contributing to the increase of $z_{u}$ accumulates to $2\delta$, then $r$ transits from $\ROOT$ to $\READY$ (Case 2).
\end{enumerate}

\MATCH
\State If there are two pending requests $r_{1}$ and $r_{2}$ such that $\ell(r_{1})=\ell(r_{2})$, match $r_{1}$ and $r_{2}$.
\State If there are two pending requests $r_{1}$ and $r_{2}$ such that $z_{\ell(r_{1})}=z_{\ell(r_{2})}=\delta$, match $r_{1}$ and $r_{2}$ and reset both $z_{\ell(r_{1})}$ and $z_{\ell(r_{2})}$ to $0$.
\State If there are two pending requests $r_{1}$ and $r_{2}$ such that $r_{1}$ is in $\READY$, match $r_{1}$ and $r_{2}$ and only reset $z_{\ell(r_{1})}$ to $0$.
\PRIORITY If multiple matches can occur at the same time, then the request that transited to state $\ROOT$ at the earliest time (as long as it exists) is prioritized.
\end{algorithmic}
\end{algorithm}

Note that the priority rule guarantees that if one request is in $\READY$, then it should be matched immediately as long as there is another pending request, which is vital in guaranteeing the $T$-impatient property.

\begin{claim}
    $\ALG$ is $4\delta$-impatient.
\end{claim}
\begin{proof}
It suffices to show that a request $r$ with arrival time $a(r)$ transits to state $\READY$ no later than $a(r)+4\delta$. It takes time at most $\delta$ to transit from $\LEAF$ to $\ROOT$ and at most $3\delta$ (the sum of Cases 1 and 2) to transit from $\ROOT$ to $\READY$ after the arrival of a request if it stays unmatched.
\end{proof}

\subparagraph{Overview of analysis.} 
Our analysis is similar to that in \cite{DBLP:conf/soda/AzarCK17,DBLP:conf/approx/AshlagiACCGKMWW17} in the sense that we analyze counters associated with different points separately, where the time axis can be divided into phases for each counter according to when it is saturated, and we compare the algorithm with another given solution based on \emph{alignment status} (to be defined later). The biggest difference is that their algorithms only match two requests on a uniform metric when both counters associated with these two requests are saturated, while our algorithm has another case that a request in $\READY$ is matched with another request immediately when it appears without waiting for another request's counter to also be saturated. This is the key property that makes \cref{alg} satisfy the desired $T$-impatient property. To deal with this case for the cost analysis, we prove that the total cost incurred by matches of this case provides a lower bound of the optimal cost as well, hence proving that our new algorithm is also competitive.

Throughout the analysis, we use $\ALG$ to denote both the algorithm that we propose and its cost, and $\SOL$ to denote a fixed offline solution and its cost that we are comparing $\ALG$ with, and the costs are always incurred under delay function $g$. The superscripts $d$ and $t$ that represent connection and delay costs respectively work in the same way as before.

We first note that our algorithm is smart (defined in \cref{section 3.1}) as well, which will be used for the alignment status analysis that we are going to introduce.
\begin{observation}\label{smart-con}
 There is at most one request pending at the same point in the metric space at every moment except when a match occurs.
\end{observation}

We can also assume that \cref{smart-con} holds for every offline solution $\SOL$, because linear delay is a special case of the size-based delay and \cref{smart} holds for this case as well. Our delay function is continuous, which means that the instantaneous delay increment at the discrete moments when a match occurs is negligible. As a result, if we want to compare the status of pending requests under $\ALG$ with another solution $\SOL$, at every moment there are four cases at every point $v\in V$: there is one request pending at $v$ under both $\ALG$ and $\SOL$, there is no request pending at $v$ under both $\ALG$ and $\SOL$, there is one request pending at $v$ under $\ALG$ but no request pending at $v$ under $\SOL$, there is one request pending at $v$ under $\SOL$ but no request pending at $v$ under $\ALG$. We say that a point $v\in V$ is \textit{aligned} if it is one of the first two cases among the four, otherwise $v$ is \textit{misaligned}. 
\begin{observation}\label{align}
    The alignment status of a point $v\in V$ changes if and only if exactly one of $\ALG$ and $\SOL$ matches $v$ with another point.
\end{observation}

We say that a request $r$ \textit{carries} $v$ if $\ALG$ makes a match between $r$ that is in $\READY$ and a request at $v\neq\ell(r)$. For such a request $r$, let $t_{1}(r)$ be the moment when it starts to be in $\ROOT$ and $t_{2}(r)$ be the moment when it starts to be in $\READY$.

\begin{claim}\label{nonoverlapping}
    Suppose that two different requests $r$ and $r'$ both have reached $\READY$ status under \cref{alg}, then $[t_{1}(r),t_{2}(r))\cap[t_{1}(r'),t_{2}(r'))=\varnothing$.
\end{claim}

\begin{proof}
    Without loss of generality, assume that $t_{1}(r)\leq t_{1}(r')$. At the time $t_{1}(r')$ when $r'$ reaches $\ROOT$, if $t_{2}(r)>t_{1}(r')$, i.e., $r$ is still unmatched in $\ROOT$ status, then $r$ and $r'$ should have been matched immediately at $t_1(r')$ if no other request matches either of them, a contradiction to that $r, r'$ both have reached $\READY$ at some later time.
\end{proof}
\begin{lemma}\label{bring}
    Suppose that a request $r$ carries some point, then $\SOL\mid_{[t_{1}(r), t_{2}(r))}\geq \delta$.
\end{lemma}

\begin{proof}
    We discuss different cases by how $r$ transits from states ROOT to READY. By the design of $\ALG$, now we cover cases 1 and 2 separately.

 \subparagraph{Case 1.}Every moment when $r$ is contributing to the increase of the root counter $z_{u}$, $r$ is the only request that is pending, which means that the total number of arrivals of requests in the input sequence so far is an odd number. Therefore, under any solution $\SOL$, there is at least one request pending and hence the delay is increasing with rate at least $1$ at every moment when $z_{u}$ is increasing. By the condition of Case 1 that $r$ has contributed $\delta$ to the increase of $z_{u}$, the delay cost of $\SOL$ during this period must be at least $\delta$ as well.
 \subparagraph{Case 2.} It is easy to see that $\ALG$ does not match any two requests at different points during $(t_{1}(r),t_{2}(r))$, since every such match must involve another request that has just reached state $\ROOT$, while the priority rule would have matched $r$ instead. If $\SOL$ makes a match between two requests at different points, then cost of at least $2\delta$ is incurred.

Otherwise, by \cref{align}, the alignment of every point stays unchanged throughout $(t_{1}(r),t_{2}(r))$. Let $\mathbf{t}=t_{2}(r)-t_{1}(r)$ denote the length of this interval, then $\mathbf{t}\geq 2\delta$ by the condition of Case 2 that $r$ transits into $\READY$. For all other points in the metric than $\ell(r)$ denoted by $v_{1},\ldots,v_{n-1}$ respectively, let $x_{1},\ldots,x_{n-1}$ denote the increased counter value of $v_{1},\ldots,v_{n-1}$ during this period respectively, then $x_{i}\leq \delta$ for every $1\leq i\leq n-1$, otherwise some counter $z_{v_{i}}$ should have been saturated since no counter is reset in this period and $r$ should have been matched according to the priority rule. Moreover, $x_{1}+\ldots+x_{n-1}\geq 2\delta$, since the total time when there is at least one counter increasing among $v_{1},\ldots,v_{n-1}$ has reached $2\delta$ for $r$ to transit into $\READY$. By the definition of the alignment status, for every $v_{i}$, either it is aligned and therefore there is a request pending at it throughout the total time $x_{i}$ when $z_{v_{i}}$ is increasing, or it is misaligned and therefore there is a request pending at it throughout the total time $\mathbf{t}-x_{i}$ when $z_{v_{i}}$ is not increasing. In either way, the delay cost incurred by a request pending at $v_{i}$ is at least $\min\{x_{i},\mathbf{t}-x_{i}\}$, which is equal to $x_{i}$ by the fact that $x_{i}\leq \delta$ and $\mathbf{t}\geq 2\delta$. As a result, the total delay time incurred at all points in $v_{1},\ldots,v_{n-1}$ is at least $x_{1}+\ldots+x_{n-1}\geq 2\delta$.

    Therefore, the cost incurred by $\SOL$ is at least $\delta$ during $[t_{1}(r), t_{2}(r))$ for both cases.
\end{proof}

The following corollary is a direct result of combining \cref{nonoverlapping} and \cref{bring}.
\begin{corollary}\label{bringup}
    Let $R=\{r\mid r\text{ carries some point during the execution of }\ALG \}$, then
    \begin{equation*}
        \SOL\geq\sum_{r\in R}\delta|R|.
    \end{equation*}
\end{corollary}

Now we start to analyze the cost of $\ALG$ under the linear delay $g(t)=t$ based on the increments of all leaf counters $z_{v}$'s and $z_{u}$. For each leaf counter $z_{v}$, let $y_{v}$ denote the total increments of the value of $z_{v}$ throughout the execution of $\ALG$ without any reset. Similarly, let $y_{u}$ denote the final value of the root counter $z_{u}$ since the value of $z_{u}$ is never reset to zero.

\begin{lemma}\label{upperbound1}
    $\ALG^{t}\leq 2\sum_{\text{leaf }v}y_{v}+y_{u}$.
\end{lemma}

\begin{proof}
    Whenever $z_{u}$ is increasing, it is the only counter that is increasing and there is exactly one request that is pending, so the total delay cost incurred while $z_{u}$ is increasing is equal to the final value of $y_{u}$.

    Whenever $z_{u}$ is not increasing, let $P_{t}$ be the set of pending requests at every such moment $t$ and $Z_t$ be the set of leaf counters that are increasing. We want to build a mapping
    $\mathcal{H}_{t}: P_{t}\rightarrow Z_t$, such that $|\mathcal{H}_{t}^{-1}(z)|\leq 2$ for every $z\in Z_t$, i.e. every unit of total leaf counter value increment can be charged to the delay cost of no more than two units. When $z_{u}$ is not increasing, no request is in state $\READY$. 
    For every request $r\in P_{t}$, if $r$ is in $\LEAF$, then $z_{\ell(r)}$ is increasing and define $\mathcal{H}_{t}(r)$ to be $z_{\ell(r)}$; if $r$ is in $\ROOT$, then there exists some other request $r'$ with $\ell(r')\neq\ell(r)$ in state $\LEAF$ whose counter $z_{\ell(r')}$ is increasing, and in this case we define $\mathcal{H}_{t}(r)$ to be $z_{\ell(r')}$. To see why $|\mathcal{H}_{t}^{-1}(z)|\leq 2$ for every $z\in Z_t$, there is at most one pending $r$ at the same location as $z$ by \cref{smart-con}, and there is at most one request that is in $\ROOT$ as well, otherwise two requests in $\ROOT$ would have been matched by $\ALG$.

    The inequality follows by combining the aforementioned arguments for two parts of the time axis together.
\end{proof}

\begin{lemma}\label{upperbound2}
    $\ALG^d\leq 2\sum_{\text{leaf }v}y_{v}$.
\end{lemma}

\begin{proof}
For every match made by $\ALG$ between two requests $r_{1}$ and $r_{2}$ at different locations, at least one counter between $z_{\ell(r_{1})}$ and $z_{\ell(r_{2})}$ is saturated at the moment when the match happens, and the saturated counter is reset to zero simultaneously with the match. Therefore, every $2\delta$ match cost can be charged to $\delta$ of the leaf counter value.
\end{proof}

Now we need to relate the final value $y_{u}$ of the root counter and $y_{v}$'s of leaf counters to the cost of any $\SOL$. Given a solution $\SOL$, for every leaf $v$, let $x_{v}$ denote the delay cost incurred by requests at $v$, $x'_{v}$ denote the connection cost incurred by matching $v$ with another point (i.e. $x'_v=k\delta$ if $v$ is matched with another point for $k$ times by $\SOL$, and the total connection cost is just the sum among all leaf points). We introduce an additional term based on $\ALG$ rather than $\SOL$, which is $b_{v}$, to denote the number of times that $v$ is \emph{carried} by another request under the execution of $\ALG$.
\begin{lemma}\label{upperbound3}
    $y_{u}\leq \SOL^{t}$.
\end{lemma}

\begin{proof}
    At every moment when $y_{u}$ is increasing, there is exactly one request pending, which means that the total amount of requests that have arrived so far is an odd number. As a result, under any solution $\SOL$, there must also be an odd number of requests pending at the moment, and the delay cost is increasing at least at a unit rate synchronously.
\end{proof}

\begin{lemma}\label{upperbound4}
    $y_{v}\leq 2x_{v}+x_{v}'+\delta \cdot b_{v}$ for every leaf $v$.
\end{lemma}

\begin{proof}
    For each leaf $v$, we divide the whole time axis into phases, where a phase ends immediately after $z_{v}$ is reset to zero, and at the same time the next phase begins. If we only consider how the states change between the beginnings of two consecutive phases, it is easy to see that $\Delta y_{v}=\delta$ for every but the last phase. We define a potential $\phi_{v}$ such that $\phi_{v}=0$ if $v$ is aligned and $\phi_{v}=\delta$ if $v$ is misaligned. Our goal is to prove that
  \begin{equation*}
        \Delta y_{v}+\Delta \phi_{v}\leq 2\Delta x_{v}+\Delta x_{v}'+\delta\cdot\Delta b_{v}
    \end{equation*}
 holds for every phase, and the result follows by summing over all phases, since both at the beginning and at the end it holds that $\phi_{v}=0$ as there is no pending request under both $\ALG$ and $\SOL$.

    For every phase but the last, it holds that $\ALG$ matches a request at $v$ with another at a different point once at the end of the phase, which flips the alignment status and changes the value of $\phi_{v}$ according to \cref{align}. However, it is possible that $\ALG$ matches $v$ with another point for more than once in total, and this can only happen if $v$ is carried by another request in $\READY$. Now we discuss different cases according to how the alignment status of $v$ changes throughout the phase, i.e. how $\ALG$ and $\SOL$ matches $v$ with a different point.

 \subparagraph{Case 1.}If $\SOL$ does not make such a match and $\ALG$ only makes one at the end of the phase, then the alignment status of $v$ stays the same until the end of the phase. If $v$ stays aligned in this period, then at every moment when $y_{v}$ is increasing along with $z_{v}$, there must be one request pending at $v$ that is incurring delay cost as well. Since $\Delta y_{v}=\delta$, it holds that $\Delta x_{v}\geq \delta$. At the end of the phase when $\ALG$ performs a match that involves a request $v$ with another not at $v$, $v$ becomes misaligned, and $\Delta\phi_{v}=\delta$ according to its definition, the inequality thus holds. If $v$ stays misaligned before the end of the phase, then $\Delta\phi_{v}=-\delta$ and the inequality holds as well.
\subparagraph{Case 2.}If $\SOL$ does not make such a match yet $\ALG$ makes more than one match between $v$ and another point before the end of the phase, then $v$ must have been carried by another request, and $b_{v}\geq 1$. If $v$ is carried only once during the phase, then $\Delta\phi_{v}=0$ as the alignment status of $v$ is flipped twice in total and $\delta=\Delta y_{v}+\Delta \phi_{v}=\delta \cdot \Delta b_{v}\leq 2\Delta x_{v}+\Delta x_{v}'+\delta\cdot\Delta b_{v}$. If $v$ is carried at least twice, then $\Delta b_{v}\geq 2$ and $\Delta y_{v}+\Delta \phi_{v}\leq 2\delta$ by the definitions of those amounts, and the inequality holds as well.
\subparagraph{Case 3.}If both $\ALG$ and $\SOL$ match $v$ with another point exactly once during the phase, then $\Delta x'_{v}=\delta$ and $\Delta \phi_{v}=0$, and the inequality follows. If $\SOL$ matches $v$ with another point exactly once but $\ALG$ makes at least two such matches, then $\Delta x'_v=\delta$ and $\Delta b_v\geq 1$, and the inequality holds because $\Delta y_v+\Delta\phi_v\leq 2\delta$. If $\SOL$ matches $v$ with another point at least twice during the phase, then $\Delta x'_{v}\geq 2\delta$ and the inequality holds as well.

    Now we only have the last incomplete phase remained. At the end of it, $\phi_{v}=0$, so $\Delta \phi_{v}$ is either $0$ or $-\delta$, while $\Delta y_{v}\leq \delta$, which implies that $\Delta y_{v}+\Delta\phi_{v}\leq 0$, and the inequality holds. This completes the proof.
\end{proof}

\begin{theorem}
    \cref{alg} achieves competitive ratio $13$ under delay function $g(t)=t$ and $13\cdot2^{k+1}((4\delta)^{k-1}+1)$ under delay function $f(t)=\frac{1}{k}t^{k}+t$ for every $k\geq 2$.
\end{theorem}

\begin{proof}
    From Lemmas~\ref{upperbound1} and~\ref{upperbound2}, the algorithmic cost $\ALG=\ALG^{t}+\ALG^{d}$ is at most $4\sum_{\text{leaf }v}y_{v}+y_{u}$. By Lemmas~\ref{upperbound3} and~\ref{upperbound4}, this is at most $\SOL^{t}+\sum_{\text{leaf }v}(8x_{v}+4x_{v}'+4\delta \cdot b_{v})=9\cdot \SOL^{t}+4\cdot \SOL^{d}+4\delta \cdot \sum_{\text{leaf }v}b_{v}$, where $b_{v}$ is the number of times that $v$ is carried by another request. By \cref{bringup}, $\delta\cdot \sum_{\text{leaf }v}b_{v}\leq \SOL$. All combined together, we have $\ALG\leq 13\cdot\SOL$, and this is the competitive ratio that $\ALG$ achieves under the linear delay function. According to \cref{impatientreduction}, it achieves competitive ratio $13\cdot2^{k+1}((4\delta)^{k-1}+t)$ under the delay function $f(t)=\frac{1}{k}t^{k}+1$ for $k>1$.
\end{proof}

\begin{remark}\label{morefunctions}
In this section, we use functions of the form $f(t)=\frac{1}{k}t^{k}+t$ for every $k>1$ as representatives as all monotone convex polynomials, and show that the competitive ratio can be preserved up to a constant factor independent of $n$ and $\delta$ because of the approximation properties as \cref{approx}. In fact, the $O(1)$-competitive ratio can be achieved by the same algorithm for a broader class of non-negative, non-decreasing convex functions where $f'(0)>0$, as long as they hold similar properties to \cref{convex} that $\frac{f(x)-f(T)}{f(x-T)}$ is upper bounded by a positive value depending on $T$ when $x\rightarrow \infty$ for any $T>0$. This covers every non-negative, non-decreasing and smooth convex function $f$ that $f'(0)>0$ whose leading term is $\Theta (x^{\alpha})$ for any $\alpha>1$, as well as $\Theta(e^{x})$. However, exceptions do exist. For example, $f(x)=e^{x^{2}}-1$ does not satisfy this property.
\end{remark}

\section{Discussions and Open Problems}

In this paper, we improved both upper and lower bounds of MPMD-Size on every $n$-point metric space. We show that the problem is equivalent to MTS on another $2^{n-1}$-point metric space, and the results are consistent with the state-of-the-art results of MTS. The deterministic competitive ratio is settled up to a constant as the $\Theta(N)$ bound of MTS holds for every $N$-point metric. However, there is still a gap between the upper and lower bounds for randomized algorithms. This is because the randomized competitive ratio of MTS depends on metric structures: as the randomized competitive ratio of $O(\log N)$ can be achieved on some $N$-point metrics like uniform metrics \cite{DBLP:conf/stoc/BorodinLS87} and weighted stars \cite{DBLP:journals/siamcomp/BubeckCLL21}, there also exist $N$-point metrics where the lower bound of $\Omega(\log ^{2}N)$ holds \cite{DBLP:conf/stoc/BubeckCR23}. As a result, to make further improvements, we need to study the structure of the metric space $(\kv, D)$. Given a metric space $(V,d)$, the definition of $(\kv, D)$ is naturally from the Boolean status of every point in $v$, with some similarities to the earthmover distance, yet this structure has not been well studied, and may be of independent interest.

Our study of MPMD-Convex on uniform metrics covers the case $f'(0)>0$ that has not been considered before, and very surprisingly, such local property at $t=0$ totally changes the competitiveness compared with the case $f'(0)=0$. As the majority of smooth convex delay functions have been covered on uniform metrics, the natural open problem is to design an algorithm for convex delays on more general metrics. There are still some open problems on uniform metrics as well: when we consider randomization, there is only a lower bound of $\Omega(\log n)$ for unbounded, non-decreasing and continuous functions that $f(0)=f'(0)=0$, while it is open whether there is a randomized algorithm that can achieve competitive ratio of $O(\log n)$ in this case. As for the case $f'(0)>0$, we note that although \cref{alg} achieves constant competitive ratio when the distance between points is a fixed constant $2\delta$, the competitive ratio depends on $\delta$ and tends to infinity when $\delta$ is sufficiently large. This is different from the linear delay case where the scale of the metric does not affect the competitive ratio. As a result, it might be interesting to study whether it is possible to design an algorithm whose competitive ratio is also independent of $\delta$.

Another natural research direction is to study the bipartite case of MPMD-Size and MPMD-Convex, where every request has either positive or negative polarity and only requests of opposite polarities can be matched. The online min-cost bipartite perfect matching with delay problem has been studied for linear~\cite{DBLP:conf/approx/AshlagiACCGKMWW17} and concave delays~\cite{DBLP:conf/soda/AzarRV21}, yet there is no result for convex or size-based delay. While the proofs and results by Deryckere and Umboh~\cite{DBLP:conf/approx/DeryckereU23} can be easily adapted for the bipartite matching problem by treating every even-sized subset whose sum of polarities is zero as an MTS state, our method for MPMD-Size does not easily extend to the bipartite case as it is not sufficient to express the states with Boolean vectors because multiple requests of the same polarity can accumulate at the same location.

\bibliography{Reference}

\end{document}